\documentclass[a4paper]{article}
\usepackage{amsmath}
\usepackage[english]{babel}
\usepackage{latexsym}
\usepackage{amssymb}
\usepackage{amscd}
\usepackage{amsgen,amstext,amsbsy,amsopn}
\usepackage{math rsfs}
\usepackage{bm,bbm}
\usepackage{amsthm,epsfig,graphicx,graphics}
\usepackage[latin1]{inputenc}
\usepackage{xspace}
\usepackage{amsxtra}
\usepackage{color}
\usepackage{dsfont}

\usepackage[pdftex]{hyperref}

\numberwithin{equation}{section}
\newcommand{\bdm}{\begin{displaymath}}
\newcommand{\edm}{\end{displaymath}}
\newcommand{\bdn}{\begin{eqnarray}}
\newcommand{\edn}{\end{eqnarray}}
\newcommand{\bay}{\begin{array}{c}}
\newcommand{\eay}{\end{array}}
\newcommand{\ben}{\begin{enumerate}}
\newcommand{\een}{\end{enumerate}}
\newcommand{\beq}{\begin{equation}}
\newcommand{\eeq}{\end{equation}}
\newcommand{\bml}[1]{\begin{multline} #1 \end{multline}}
\newcommand{\bmln}[1]{\begin{multline*} #1 \end{multline*}}

\newcommand{\lf}{\left}
\newcommand{\ri}{\right}

\newcommand{\xv}{\mathbf{x}}

\newcommand{\rv}{\mathbf{r}}

\newcommand{\diff}{\mathrm{d}}
\newcommand{\eps}{\varepsilon}

\newcommand{\avi}{\mathbf{a}_i}
\newcommand{\avik}{\mathbf{a}_{i_k}}
\newcommand{\avj}{\mathbf{a}_j}
\newcommand{\avjh}{\mathbf{a}_{j_h}}
\newcommand{\jv}{\mathbf{j}}

\newcommand{\gpf}{\mathcal{E}^{\mathrm{GP}}}
\newcommand{\gpe}{E^{\mathrm{GP}}}
\newcommand{\gpm}{\Psi^{\mathrm{GP}}}

\newcommand{\hgpf}{\hat{\mathcal{E}}^{\mathrm{GP}}}
\newcommand{\hgpe}{\hat{E}^{\mathrm{GP}}}
\newcommand{\hchem}{\hat{\lambda}^{\mathrm{GP}}}

\newcommand{\dg}{\mbox{deg}}
\newcommand{\curl}{\mbox{curl}}

\newcommand{\set}{\mathcal{S}}

\newcommand{\tff}{\mathcal{E}^{\mathrm{TF}}}
\newcommand{\tfd}{\mathcal{D}^{\mathrm{TF}}}
\newcommand{\tfe}{E^{\mathrm{TF}}}
\newcommand{\tfm}{{\rho^{\mathrm{TF}}}}
\newcommand{\rtf}{{R^{\mathrm{TF}}}}

\newcommand{\rbl}{R_>}
\newcommand{\rbulk}{R_{\rm bulk}}
\newcommand{\cbulk}{C_{\rm bulk}}

\newcommand{\tfchem}{\lambda^{\mathrm{TF}}}

\newcommand{\trial}{\Psi_{\mathrm{trial}}}

\newcommand{\cell}{\mathcal{Q}}
\newcommand{\celli}{\mathcal{Q}^i_{}}

\newcommand{\half}{\frac{1}{2}}

\newcommand{\G}{\mathcal{G}}

\newcommand{\yv}{\mathbf{y}}

\newcommand{\gpdom}{\mathscr{D}^{\mathrm{GP}}}

\newcommand{\hgpdom}{\hat{\mathscr{D}}^{\mathrm{GP}}}

\newcommand{\Ofirst}{\Omega_{\mathrm{c_1}}}

\newcommand{\tx}{\textstyle}

\newcommand{\vtrial}{v_{\mathrm{trial}}}
\newcommand{\phitrial}{\phi_{\mathrm{trial}}}
\newcommand{\nutrial}{\nu_{\mathrm{trial}}}
\newcommand{\neps}{N_{\eps}}
\newcommand{\ai}{\mathbf{a}_i}
\newcommand{\aj}{\mathbf{a}_j}
\newcommand{\xivor}{\xi_{\mathrm{v}}}
\newcommand{\ceps}{c_{\eps}}

\newcommand{\kk}{\mathcal{K}}

\newcommand{\hrho}{H_{\rho}}
\newcommand{\irho}{\I_{\rho}}
\newcommand{\irhoe}{I_{\rho}}

\newcommand{\rrs}{R_{\star}}
\newcommand{\ms}{m_{\star}}

\newcommand{\rsm}{R_<}

\newcommand{\R}{\mathbb{R}}
\newcommand{\N}{\mathbb{N}}

\newcommand{\C}{\mathbb{C}}

\newcommand{\D}{\mathcal{D}}
\newcommand{\F}{\mathcal{F}}
\newcommand{\E}{\mathcal{E}}

\newcommand{\B}{\mathcal{B}}

\newcommand{\OO}{\mathcal{O}}
\newcommand{\I}{\mathcal{I}}

\newcommand{\ep}{\varepsilon}

\newcommand{\Om}{\Omega}

\newcommand{\dd}{\partial}

\newcommand{\jt}{\tilde{\textbf{\j}}}
\newcommand{\mut}{\tilde{\mu}}

\newcommand{\musta}{\mu_{\star}}
\newcommand{\mustar}{\mu_{\rho}}

\newcommand{\chiin}{\chi_{\mathrm{in}}}
\newcommand{\chiout}{\chi_{\mathrm{out}}}

\newcommand{\one}{\mathds{1}}

\newcommand{\hstar}{h_{\mustar}}
\newcommand{\mutrial}{\mu_{\rm trial}}

\newcommand{\M}{\mathscr{M}}
\newcommand{\supp}{\mathrm{supp}}

\newcommand{\nablap}{\nabla^{\perp}}
\newcommand{\rvp}{\rv^{\perp}}

\newcommand{\eg}{e_g}

\newcommand{\tfr}{R ^{\rm TF}}
\newcommand{\Rc}{R_{\rm cut}}
\newcommand{\tfI}{I ^{\rm TF}}
\newcommand{\tfIc}{\I ^{\rm TF}}
\newcommand{\tfpot}{F ^{\rm TF}}
\newcommand{\tfH}{H ^{\rm TF}}
\newcommand{\tfHp}{H ^{\rm TF '}}

\newcommand{\htilde}{\tilde{h}_{\nutrial}}

\newtheorem{teo}{Theorem}[section]
\newtheorem{lem}{Lemma}[section]
\newtheorem{pro}{Proposition}[section]
\newtheorem{cor}{Corollary}[section]

\newcounter{remark}[section]
\newenvironment{rem}{\refstepcounter{remark} \vspace{0,1cm} \noindent \textit{Remark \thesection.\theremark}\,}{\\\hfill \qed }

\begin{document}

\markboth{\scriptsize{\textsc{Correggi, Rougerie} -- Inhomogeneous Vortex Patterns in BECs}}{\scriptsize{\textsc{Correggi, Rougerie} -- Inhomogeneous Vortex Patterns in BECs}}

\title{Inhomogeneous Vortex Patterns in Rotating Bose-Einstein Condensates}

\author{ M. Correggi${}^{a}$, N. Rougerie${}^{b}$
	\\
	\normalsize\it ${}^{a}$ Dipartimento di Matematica, Universit\`{a} degli Studi Roma Tre,	\\
	\normalsize\it L.go S. Leonardo Murialdo 1, 00146, Rome, Italy.	\\
	\normalsize\it ${}^{b}$ Universit\'e de Grenoble 1 and CNRS, LPMMC \\ 
	\normalsize\it Maison des Magist\`{e}res CNRS, BP166, 38042 Grenoble Cedex, France.}
	
\date{September 19, 2012}

\maketitle

\begin{abstract} 
We consider a 2D  rotating Bose gas described by the Gross-Pitaevskii (GP)
theory and investigate the properties of the ground state of the theory
for rotational speeds close to the critical speed for vortex nucleation.
While one could expect that the vortex distribution should be
homogeneous within the condensate we prove by means of an
asymptotic analysis in the strongly interacting (Thomas-Fermi) regime  that it is not.
More precisely we rigorously derive a formula due to Sheehy and
Radzihovsky [{\it Phys. Rev. A} {\bf 70}, 063620(R) (2004)] for the vortex distribution, a consequence of which is that
the vortex distribution is strongly inhomogeneous close to the critical
speed and gradually homogenizes when the rotation speed is increased.

From the mathematical point of view, a novelty of our approach is that we
do not use any compactness argument in the proof, but instead provide
explicit estimates on the difference between the vorticity measure of the
GP ground state and the minimizer of a certain renormalized energy functional.

	\vspace{0,2cm}

	MSC: 35Q55,47J30,76M23. PACS: 03.75.Hh, 47.32.-y, 47.37.+q.
	\vspace{0,2cm}
	
	Keywords: Bose-Einstein condensates, quantized vortices, Gross-Pitaevskii energy, renormalized energy, vortex lattices.
\end{abstract}

\tableofcontents

\section{Introduction}

Vortex nucleation in equilibrium states of rotating fluids is a hallmark of superfluidity and its experimental observation in rotating Bose gases was a key step in the understanding of the properties of Bose-Einstein Condensates (BECs). Among other spectacular observations, that of triangular lattices\footnote{`Abrikosov lattices' of vortices analogous to those occurring in type-II superconductors.} containing up to hundreds of vortices (see, e.g., \cite{BSSD,ARVK,MCWD,RAVXK,CHES}) gave a new strong motivation for theoretical studies. 

A rather natural question, that has not been addressed immediately after the observation of vortex lattices, can be formulated as follows. The BECs that are produced in laboratories are trapped gases, meaning that the confinement against centrifugal forces is provided by a magneto-optical potential. As a consequence the matter density profile of the condensate is not homogeneous and depends strongly on the type of trap that is being used. In this respect, it is rather striking that the vortex lattices observed in experiments seem to be perfectly homogeneous with a uniform mean distribution of vortices in the samples. Indeed, the vortex density could depend on the underlying matter density in various ways, for example regions of low matter density such as the boundary of the fluid could attract and pin the vortices. How come that such effects do not seem to be observed and that the vortex lattices are homogeneous, at least to a very good approximation?

To our knowledge, this question has been formulated and answered first by Sheehy and Radzihovsky in \cite{SR1,SR2}. More precisely, the relation between the matter density and the vortex density has been elucidated based on formal arguments (see also \cite{BPGW}), leading to a formula whose efficiency has been favorably compared to experimental data \cite{CHES} and to numerical simulations \cite{Dan}. Among the findings of \cite{SR1,SR2} is the fact that the vortex density {\it does} depend on the matter density, but in a subtle way that has leading order effect only close to the critical speed for vortex nucleation. It is thus not surprising that the vortex lattices seem completely homogeneous in experimental situations since, in order to observe many vortices, the rotational velocity is taken well above the first critical speed. However, some slight inhomogeneity of the vortex lattice survives for these large angular velocities, as predicted by \cite{SR1,SR2}, and it can in fact be observed as a small correction to almost uniform vortex densities \cite{CHES,Dan}. 

\medskip

Although it is experimentally difficult to observe the transition regime where the vortex lattice is expected to be inhomogeneous, the theoretical question remains of interest. The main model used for the description of rotating BECs is the so-called Gross-Pitaevskii theory, which can be rigorously derived from the underlying many-body problem (see \cite{LSSY,LS} and references therein) in a suitable limit. It is of importance to be able to rigorously derive a formula for the vortex density from GP theory.

In this context, the critical speed for vortex nucleation has been rigorously computed in  \cite{IM1,AJR} and the distribution of the first few vortices to appear in the condensate studied in \cite{IM2}. On the other hand, the regime well above the first critical speed has been treated in \cite{CY,CPRY1,CPRY3} where it has been shown that the vortex density is homogeneous to leading order, in the sense that many vortices are packed in the condensate and their average distribution is uniform. Note that the latter contributions concern the average distribution of vortices in a regime where they are densely packed in the fluid, it does not give access to the precise pattern formed by the vortices. A rigorous proof starting from GP theory that the vortices arrange on a triangular lattice seems to be still out of reach, despite recent advances in the related Ginzburg-Landau (GL) theory \cite{SS4}.

Summing up, there is still a gap in the rigorous theory of BECs between \cite{IM1,IM2,AJR} and \cite{CY,CPRY1,CPRY3}. This gap corresponds to the regime where the rotation speed is larger than the first critical speed but of the same order of magnitude and this is precisely the regime where the inhomogeneity of the (average) vortex distribution should come into play. The present paper aims at filling this gap by rigorously deriving from GP theory the formula for the vortex distribution of \cite{SR1,SR2}.

\medskip

Our mathematical setting is the following: we consider a two-dimensional rotating BEC confined by a trapping potential $ V(r) = r^{s} $ (with\footnote{Throughout the paper we use the following convention: a vector will be denoted in bold fonts (e.g., $ \mathbf{r} $), whereas normal fonts will always denote scalars (e.g., $ r = |\mathbf{r}| $), which might also be the modulus of a vector.} $ r = |\rv| $), in the framework of the GP theory. After a suitable scaling of length units (see \cite[Section 1.1]{CPRY3}), the GP energy functional can be written
\beq
	\label{gpf}
	\gpf[\Psi] : = \int_{\R^2} \diff \rv \: \bigg\{ \frac{1}{2} \lf| \nabla \Psi \ri|^2 - \Omega \Psi^* L \Psi + \frac{V(r)}{\eps^2} |\Psi| ^2+ \frac{\lf| \Psi \ri|^4}{\eps^2} \bigg\},
\eeq
where $ \Omega $ is the angular velocity,
\beq
	\label{ext pot}
	V(r) : = r^s,	\hspace{1cm}	s \geq 2,
\eeq 
$ L $ stands for the third component of the angular momentum, i.e., in polar coordinates $ \rv = (r ,\vartheta ) $, $ L = - i \partial_{\vartheta} $ or equivalently $ L = \rv \cdot \nablap $, $ \nablap : = (-\partial_y, \partial_x) $, and the coupling parameter $ \eps >0 $ is going to be assumed small ($ \eps \ll 1 $), i.e., we study the so called Thomas-Fermi (TF) limit of strong interactions. The wave function $ \Psi: \R ^2 \mapsto \C$ belongs to the domain
\beq
	\gpdom : = \lf\{ \Psi \in H^1(\R^2) \cap L^4(\R^2) : \: r^s |\Psi|^2 \in L^1(\R^2), \lf\| \Psi \ri\|_{L ^2(\R ^2)} = 1 \ri\}
\eeq
 and the ground state energy of the system is obtained by the minimization of $ \gpf$:
\beq
	\label{gpe}
	\gpe : = \inf_{\Psi \in \gpdom} \gpf[\Psi].
\eeq
We denote by $ \gpm $ any associated minimizer (there is no uniqueness in general). 

We will mostly be interested in spotting the vortices of $\gpm$, that is its zeros that carry a non-zero phase circulation (i.e., non-trivial topological degree or winding number). In particular, connecting to the preceding discussion, we would like to derive a relation between the distribution of vortices and the matter density of the system. The latter, given by $|\gpm| ^2$, can be approximated by minimizing the simplified TF functional
\beq
	\label{tff}
	\tff[\rho] : = \eps^{-2} \int_{\R^2} \diff \rv \: \lf[ r^s + \rho \ri] \rho,
\eeq
obtained by dropping the kinetic terms in \eqref{gpf}. Here $\rho\geq 0$ plays the role of the matter density, normalized such that $\int_{\R ^2} \rho= 1$. The minimizer of \eqref{tff} is the explicit radial function 
\beq
	\label{tfm}
	\tfm(r) = \half \lf[\tfchem - r^s\ri]_+,
\eeq
where $ [ \: \cdot \: ]_+ $ stands for the positive part and $ \tfchem $ is a normalization parameter that ensures $ \lf\| \tfm \ri\|_1 = 1 $. Note that $ \tfm $ has compact support in the ball of radius $ \tfr $ and a rather simple computation yields
\beq
	\label{tfchem}
	\tfe = \frac{\pi s}{4(s+1) \eps^2} \lf(\tfchem\ri)^{2(s+1)/s},	\hspace{1cm}	\rtf = \lf(\tfchem\ri)^{1/s},	\hspace{1cm}	\tfchem = \lf( \frac{2(s+2)}{\pi s} \ri)^{s/(s+2)},
\eeq
with $ \tfe $ standing for the TF ground state energy. Essentially, $\tfd$ is the region occupied by the condensate: we are able to prove that the mass contained in $\R ^2 \setminus \tfd$ is extremely small, the reason being that $\gpm$ decays exponentially (both as a function of $\rv$ and $\ep$) in this region. 

To spot the vortices of $\gpm$ in the bulk $\tfd$ of the condensate, it is standard to consider the so-called \emph{vorticity measure} $\mu$ associated with $\gpm$. As we are dealing with a regime where the condensate contains a large (actually $\propto |\log \ep|$) number of vortices, $\mu$ should be interpreted as giving the mean distribution of vortices. With the definition we will adopt below (see \eqref{eq:intro mu}), one can actually prove that, if $\gpm$ contains $J$ vortices of degrees $d_1,\ldots, d_J$  and locations $\mathbf{a}_1,\ldots,\mathbf{a}_J$ 
\begin{equation}\label{eq:intro vortic}
\mu \approx 2\pi |\log \ep| ^{-1} \sum_{j=1} ^J  d_j \delta_{\mathbf{a}_j} 
\end{equation}
in the $(C^1_c) ^*$ topology, where $\delta_{\mathbf{a}_j}$ stands for the Dirac mass at $\mathbf{a}_j$. Note the scaling factor $|\log \ep| ^{-1}$ in \eqref{eq:intro vortic}: since the condensate contains $\OO (|\log \ep|)$ vortices, $\mu$ will be a quantity of order $1$.

\medskip

The mechanism for vortex nucleation in rotating superfluids is now well understood, see, e.g., \cite{AAB,IM1,IM2,CRY}. A vortex becomes favorable in the system if it can lower the energy by interacting with the rotation field. More precisely, the interaction with the rotation field should overcome the energetic cost for vortex nucleation, given by
\begin{equation}\label{eq:intro vcost}
\pi |\log \ep| |d_j|\tfm (a_j)
\end{equation}
with $d_j$ the degree of the vortex and $\mathbf{a}_j$ its location. To evaluate the energy gain brought by the vortex compensating the rotation, it is convenient to introduce the potential function
\begin{equation}\label{eq:intro pot}
\tfpot (r) = - \frac{\Omega}{|\log \ep|} \int_{r} ^{\tfr} \diff t \: t \: \tfm(t).
\end{equation}
The energetic gain of a vortex is then 
\begin{equation}\label{eq:intro vgain}
2 \pi |\log \ep| d_j F(a_j). 
\end{equation}
By comparing the cost and gain of vortex nucleation it is obvious that in order to be energetically favorable, vortices should have positive degrees $d_j$, because $\tfpot$ is clearly negative. These considerations lead to the definition of a radial function giving the energetic cost of a vortex of degree $1$ located at $\avj$ with $|\avj| = r$:
\begin{equation}\label{eq:into costf}
\tfH (r)= \half \tfm(r) + \tfpot (r). 
\end{equation}
Looking for its minimum, one finds that it lies at $r=0$, indicating that this is where a vortex is most favorable. Equating gain and cost of a vortex at the origin, one can see that the critical speed for vortex nucleation is given by 
\beq
	\label{eq:first critical speed}
	\Ofirst = \Omega_1 |\log\eps|,	\hspace{1cm}	\Omega_1 : =  \frac{\pi}{2} \lf( \frac{2(s+2)}{\pi s} \ri)^{s/(s+2)},
\eeq
namely $\tfH(0) > 0$ for $\Om < \Ofirst $  and $\tfH(0) < 0$ for $\Om > \Ofirst$. Note that the above value is that one finds by applying the rigorous analysis of \cite{IM1,AJR}. 

The question we address in this paper is what happens when $\Omega$ is chosen of the form
\begin{equation}\label{eq:intro Omega}
\Om = \Om_0 |\log \ep|, \qquad \Om_0 > \Om_1,\qquad \Om_0 = \OO(1), 
\end{equation}
i.e., when $\Omega$ is strictly larger than the first critical speed but of the same order of magnitude when $\ep \to 0$. We prove that in this regime the vorticity measure $\mu$ is to leading order\footnote{We denote by $ \one_{\set} $ the characteristic function of the set $ \set $.}
\begin{equation}\label{eq:intro result}
\mu \approx \left[\nabla \lf( \frac{1}{\tfm} \nabla \tfH \ri) \right]_+ \one_{ \left\lbrace \tfH \leq 0\right\rbrace}
\end{equation}
in a sense to be made precise in the next section. Computing the term $\nabla \frac{1}{\tfm} \nabla \tfH$  this can be rewritten as 
\begin{equation}
\label{eq:intro result 2}
\mu(r) \approx \left[  \frac{1}{2}   \partial_r ^2 \log \lf( \tfm(r) \ri) + 2 \Omega_0 \right] _{+} \one_{\left \lbrace\tfH \leq 0\right\rbrace},
\end{equation}
which is the formula found by Sheehy and Radzihovsky, except that our analysis shows that vortices should lie only where the cost function is negative, which was not clearly mentioned in \cite{SR2}. Note that vortices are also confined to the region where 
\[
 \half \partial_r  ^2 \log \lf( \tfm(r) \ri) + 2 \Omega_0 \geq 0,
\]
which is far from obvious if only the cost and gain considerations sketched above are taken into account. Indeed, the particular form \eqref{eq:intro result 2} is to a large extent due to the interaction between vortices. 

\medskip

Several interesting properties of the (average) vortex distribution as a function of $\Omega$ can be read off from \eqref{eq:intro result 2}:
\begin{itemize}
\item When $\Omega_0 < \Omega_1$ in \eqref{eq:first critical speed}, a straightforward computation reveals that the cost function $\tfH$ is positive everywhere. Vortices are thus not favorable and the vortex distribution vanishes identically. We thus recover the expression of \cite{IM1,AJR} for the first critical speed.
\item In the regime \eqref{eq:intro Omega}, one can compute (see Section \ref{sec:limit problem}) that there is a non-empty region where $\tfH <0$ and $\nabla \frac{1}{\tfm} \nabla \tfH > 0$, whose size increases with increasing $\Omega_0$ until it finally fills the whole sample in the limit $\Omega_0 \to \infty$. This region is, according to \eqref{eq:intro result}, filled with vortices, a behavior that is reminiscent of the `obstacle problem regime' in GL theory \cite{SS1}. 
\item The first term in the right-hand side of \eqref{eq:intro result 2} can not be constant, except when $\tfm$ is constant itself, which can only happen in the somewhat unrealistic case of the flat trap considered, e.g., in \cite{CPRY1}. This term is thus responsible for an inhomogeneity of the vortex distribution, whereas the second term yields a constant contribution of $2 \Omega_0$ ($2 \Omega$ in the physical variables) units of vorticity per unit area.
\item The first term in \eqref{eq:intro result 2}, responsible for the inhomogeneity, is independent of $\Omega_0$. Its importance relative to the second one thus diminishes with increasing $\Omega_0$. The inhomogeneity then becomes a second order correction in the limit $\Omega_0 \to \infty$, which corresponds to $\Omega \gg |\log \ep|$. As discussed above, this is observed in experiments and numerical simulations. This also bridges with the situation considered in \cite{CPRY3}, where we proved that, if $\Omega \gg |\log \ep|$, the vortex density is to leading order constant and proportional to $2\Omega$.
\end{itemize}

The next section is devoted to a more precise statement and discussion of our results. Proofs are given in Sections \ref{sec:limit problem}, \ref{sec:upper bound} and \ref{sec:lower bound}. 

\bigskip

\noindent\textbf{Acknowledgements.}
It is always a pleasure to acknowledge the hospitality of the \textit{Erwin Schr\"odinger Institute for Mathematical Physics} where part of this research has been carried out. M.C. acknowledges the support of the European Research Council under the European Community Seventh Framework Program (FP7/2007-2013 Grant Agreement CoMBos No. 239694). N.R wishes to thank Alessandro Giuliani and Weizhu Bao for hospitality, respectively at the {\it Universit\`{a} di Roma Tre} and at the \textit{Institute for Mathematical Sciences} in Singapore. Stimulating discussions with Jakob Yngvason were also much appreciated.

\section{Statement of the Main Results}\label{sec:results}

We now turn to a more precise description of our results. Recall that we consider the asymptotic behavior of the ground state energy and minimizer of the functional \eqref{gpf} when $\ep \to 0$ with the scaling
\beq
	\label{Omega 2}
	\Omega = \Omega_0 |\log\eps|,	\hspace{1cm}	\Omega_0 > \Omega_1,
\eeq
where $ \Omega_0 $ is assumed to be constant for the sake of simplicity and $\Om_1$ is defined in \eqref{eq:first critical speed}. The regime $\Omega_0\gg 1$ corresponds to what we have considered in \cite{CPRY3}, while in \cite{IM1,IM2} it was studied the case $\Om = \Om_1 |\log \ep| + \OO(\log |\log \ep|)$.   

To take into account the inhomogeneous density profile it is convenient to introduce the following energy functional
\beq
	\label{hgpf}
	\hgpf[f] : = \int_{\R^2} \diff \rv \: \lf\{ \half |\nabla f |^2 + \eps^{-2} \lf[ r^s + f^2 \ri] f^2 \ri\}
\eeq
with ground state energy 
\beq
	\label{hgpdom}
	\hgpe : = \min_{f \in \hgpdom} \hgpf[f],	\hspace{1cm}	\hgpdom : = \lf\{ f \in \gpdom : f = f^* \ri\}.
\eeq
Standard arguments show that there is a unique strictly positive radial minimizer which will be denoted by $ g $. Note that $\hgpf$ coincides with $\gpf$ restricted to \emph{real} functions. In this sense, its minimization corresponds to the search for a vortex-free profile. The main difference between $\hgpf$ and $\tff$ is that the former includes the contribution of the radial kinetic energy due to the bending of the density profile. We will prove in the Appendix that $g^2 \approx \tfm$ in a suitable sense.

\medskip

The expression \eqref{eq:intro result} enters our problem through the minimization of a `renormalized energy' (we employ the consecrated terminology of GL theory \cite{BBH,SS2}) expressing the energy of a given vorticity measure $\nu$ in terms of the TF density $\tfm$ in units of $ |\log\eps|^2 $
\beq
	\label{ren energy}
	\tfIc [\nu] = \int_{\tfd} \left\{ \frac{1}{2\tfm} |\nabla h_{\nu}|^2 + \frac{1}{2} \tfm |\nu| + \tfpot \nu \right\} ,
\eeq
where\footnote{In all the paper, $\B(\varrho)$ denotes the disc centered at the origin with radius $ \varrho $, while $B(\rv,\varrho)$ is the same disc but centered at $\rv \in \R ^2$. Also we will sometimes omit the measure in the integral (as, e.g., in the first term of \eqref{ren energy}), when it is the usual two-dimensional Lebesgue measure $ \diff \rv $.}
\beq
	\label{domain D}
	\tfd : = \supp (\tfm) = \B(\tfr),
\eeq
and
\begin{equation}
	\label{eq ren energy}
	 \begin{cases}
		-\nabla \left( \frac{1}{\tfm} \nabla h_{\nu}\right) = \nu	 \mbox{ in } \tfd, \\
		h_{\nu} = 0 								 \mbox{ on } \partial \tfd.
	\end{cases} 
\end{equation}
Recall that $\tfpot$ is defined in \eqref{eq:intro pot}. The minimization of the renormalized energy in its natural energy space
\beq
	\label{eq: intro measure class}
	\M_{\tfm} (\tfd) = \left\lbrace \nu \in \left(C^0_c (\tfd)\right) ^*,\: \int_{\tfd} \lf\{ \frac{1}{\tfm} |\nabla h_{\nu}| ^2 +  \tfm |\nu| \ri\} < +\infty \right\rbrace                                                                                                                                        \eeq
is discussed in details in Section \ref{sec:limit problem}. We prove (see Theorem \ref{theo:renorm energy}) that $ \tfIc $ has a unique minimizer $ \musta $ among the measures in $\M_{\tfm} (\tfd)$. It is explicitly given by 
\beq
	\label{musta}
	\musta =  \lf[ \nabla \lf( \frac{1}{\tfm} \nabla \tfH \ri) \ri]_+ \one_{\left \lbrace\tfH \leq 0\right\rbrace},
\eeq
i.e., it is exactly the measure appearing in the right-hand side of \eqref{eq:intro result}. We also set 
\begin{equation}\label{eq:tfI}
\tfI: = \tfIc [\musta] = \frac{1}{2} \int_{\supp(\musta)} \tfH \musta
\end{equation}
where the second equality is proved in Section \ref{sec:limit problem} below. Note that by \eqref{musta} one easily has 
\beq
	 \tfI \leq 0,
\eeq
since $ \musta \geq 0 $ and $ \tfH \leq 0  $ on the support of $ \musta $. Moreover both the renormalized energy as well as its minimizer $ \musta $ are fixed when $\ep\to 0$, thanks to the extraction of a scaling factor $ |\log\eps|^2 $.

We can now formulate our first result about the GP ground state energy asymptotics:
\begin{teo}[\textbf{Ground state energy asymptotics}]
	\label{teo:gse asympt}
	\mbox{}	\\
	If $ \Omega = \Omega_0 |\log\eps| $, with $ \Omega_0 > \Omega_1 $, then
	\beq
		\label{gse asympt}
		\gpe = \hgpe +\tfI  |\log\eps|^2 \left( 1 + \OO\left(\frac{\log |\log \ep|}{|\log\eps| ^{1/2}}\right)\right)
	\eeq
	in the limit $\ep \to 0$.
\end{teo}

\begin{rem}(Composition of the ground state energy)\mbox{}	\\
The leading order contribution to $ \gpe $ is given by $\hgpe$, which in turn contains $\tfe$ of order $\ep ^{-2}$ (see \eqref{tfchem}) and a remainder due to the radial kinetic energy of the vortex-free profile. Such a correction can be shown to be of order $ |\log\eps| $ \cite[Proposition 2.1]{CPRY3}, i.e., much smaller than the contribution of vortices $ \tfI |\log\ep| ^2$, which on the other hand is a rather small correction to the main term $ \tfe $. 
\end{rem}

\begin{rem}(The renormalized energy)\mbox{}\label{rem:ren energy}\\
Note that the functional \eqref{ren energy} is defined for the `reasonable' vorticity measures in $\M_{\tfm} (\tfd)$, in particular those arising from wave functions in the manner of \eqref{eq:intro mu} below. It is not well-defined for sums of Dirac masses, which is a significant difficulty in the analysis, in particular in view of \eqref{eq:intro vortic}. The first term in \eqref{ren energy} corresponds to the interaction between vortices, it is computed in a way reminiscent of electrostatics: $h_{\nu}$ is similar to a potential generated by individual electric charges distributed according to the charge density $ \nu $ and the first term in \eqref{ren energy} is the corresponding electrostatic energy. The other two terms can be understood as the sum of the cost and gain due to each individual vortex, in view of \eqref{eq:intro vortic}, \eqref{eq:intro vcost} and \eqref{eq:intro vgain}. Note that, if it was a priori known that the minimization could be restricted to positive measures and the density $\tfm$ was constant, \eqref{ren energy} would reduce exactly to the electrostatic energy of a positive charge distribution in the potential $\half \tfm + \tfpot$. In this case  \eqref{eq ren energy} would become the Poisson equation for the charge distribution $ \nu $: $h_{\nu}$ could then be interpreted as an electrostatic potential and its gradient as the corresponding field. In this analogy, the non-constant weight $ 1/\tfm $ can be thought of as modeling a sample with non-homogeneous conductivity.    
\end{rem}

\medskip

As we will prove below, the minimization of $\hgpf$ gives the matter density of the system to a very good approximation, i.e., $|\gpm| ^2 \approx g^2$. It is thus natural to write $\gpm$ in the form
\begin{equation}\label{eq:decouple formal}
\gpm = g u,  
\end{equation}
where $u$ is essentially a phase factor accounting for the vortices of $\gpm$, times a profile vanishing close to the vortex cores and almost equal to $1$ elsewhere. A convenient way of spotting the vortices contained in $u$ is to use the so-called \emph{vorticity measure}
\begin{equation}\label{eq:intro mu}
\mu:= |\log \ep| ^{-1} \curl \left[ \frac{i}{2}\left(u\nabla u ^* - u ^* \nabla u \right) \right]
\end{equation}
which is (up to the $|\log \ep| ^{-1}$ factor) nothing but the curl of the superfluid current 
\begin{equation}\label{eq:intro j}
\mathbf{j} :=  \frac{i}{2}\left(u\nabla u ^* - u ^* \nabla u \right)
\end{equation}
and thus (in analogy with fluid mechanics) a good candidate to count the vortices of $u$.

As usual \cite{AAB,IM2,R}, energy methods do not allow to spot vortices lying too close to the boundary of the domain. Indeed, $\tfm$ vanishes on $\dd \tfd$ and thus, according to \eqref{eq:intro vcost}, a vortex lying close to the boundary carries very little energy. We will thus limit ourselves to analyze the behavior of $\mu$ in the smaller ball $\B(\rbulk)$ with $\rbulk$ satisfying 
\begin{equation}\label{eq:intro Rc}
\rbulk < \tfr,\qquad \left| \rbulk - \tfr\right| = \OO(\Om ^{-1}).
\end{equation}
In fact in \eqref{eq:low radii} below we will make the precise choice $ \rbulk = \rtf - C \Omega^{-1} $ for some given explicit constant $ C $. However our proof works just the same provided the fixed constant $C$ is chosen small enough.

Note that by restricting ourselves to $\B(\rbulk)$ we are only neglecting a small part of $\tfd$ and the domain $\B(\rbulk)$ contains the bulk of the mass of the condensate, in the sense that\footnote{This is an easy consequence of the fact that $|\gpm|$ is uniformly bounded (see Appendix).}
\begin{equation}\label{eq:mass bulk}
\int_{\B(\rbulk)} \diff \rv \: |\gpm| ^2= 1 - o(1) 
\end{equation}
in the limit $\ep \to 0$.

We prove that $\mu$ is close to $\musta$ in $\B(\rbulk)$, with the meaning of `close to' specified by the following norm, which is defined for measures $\nu$:
\begin{equation}\label{eq:norm}
\| \nu \|_{\tfm}:= \sup_{\phi \in C^1_c (\B(\rbulk))} \frac{\bigg| \displaystyle\int_{\B(\rbulk)} \nu \phi \bigg|}{\displaystyle \left(\int_{\B(\rbulk)} \diff \rv \: \frac{1}{\tfm} |\nabla \phi| ^2\right)^{1/2} + \| \nabla \phi\|_{L ^{\infty} (\B(\rbulk))}}. 
\end{equation}

\begin{teo}[\textbf{Asymptotics for the vorticity measure}]\label{teo:vorticity}\mbox{}\\
Let $\mu$ and $\musta$ be defined respectively in \eqref{eq:intro mu} and \eqref{musta}. Then we have
\begin{equation}\label{eq:vortic asympt}
\left\| \mu - \musta \right\|_{\tfm} \leq \OO \left( \frac{\log |\log \ep| ^{1/2} }{|\log \ep| ^{1/4}}\right)
\end{equation}
in the limit $\ep \to 0$.
\end{teo}

\begin{rem}(The norm $ \lf\| \: \cdot \: \ri\|_{\tfm} $)\label{rem:norm}\mbox{}\\
The expression \eqref{eq:norm} is a rather natural definition for the norm of a measure. The norm of the test function $\phi$ appearing in the denominator contains two parts. The first is naturally associated with the energy functional \eqref{eq:tfI}, so an optimal statement must necessarily include this term. The second term  $\| \nabla \phi\|_{L ^{\infty} (\B(\rbulk))}$ appears when regularizing $\mu$ in the course of the proof. Without this term, the norm $\| \! \cdot \! \|_{\tfm}$ would not be well-defined for a Dirac mass, and in view of \eqref{eq:intro vortic} this would be rather problematic in our setting. As we will see in the proof, a suitable regularization of $\mu$ can be estimated in the norm where the term $\| \nabla \phi\|_{L ^{\infty} (\B(\rbulk))}$ is removed in the denominator of \eqref{eq:norm}. We still believe that it is necessary to include that term to state a result on the asymptotics of $\mu$. 

Note finally that if $\tfm$ was uniformly bounded below by a positive constant independent of $ \eps $ in $\B(\rbulk)$ (which is not the case), the norm \eqref{eq:norm} would be equivalent to the $\left(C^1_c (\B(\rbulk))\right)^*$ norm. In particular, this means that for any \emph{fixed} $R<\tfr$ and $\ep$ small enough, \eqref{eq:vortic asympt} yields an estimate on the $\left(C^1_c (\B(R))\right)^*$ norm of $ \mu - \musta $. Indeed, for $\ep$ small enough, $\B(R)\subset \B(\rbulk)$ and $\tfm$ is bounded below in $\B(R)$. We can thus deduce from Theorem \ref{teo:vorticity} that for any fixed $R<\tfr$
\[
 \mu \rightarrow \musta
\]
strongly in $\left(C^1_c (\B(R))\right)^*$ as $\ep \to 0$. The reason why we consider the approximation of $\mu$ by $\musta$ in the larger ball $\B(\rbulk)$ has already been explained: such a ball contains the bulk of the mass in the $\ep \to 0$ limit, whereas $\B(R)$ does not.
\end{rem}

\begin{rem}(Asymptotics for an explicit vorticity measure)\label{rem:explicit}\mbox{}\\
It is also possible to state the result in terms of an `explicit' vorticity measure of the form of the right-hand side of \eqref{eq:intro vortic}, in the spirit of, e.g., \cite[Theorem 1.2]{R} or \cite[Theorem 1.1]{CPRY1}. Indeed, the proof of Theorem \ref{teo:vorticity} requires to localize the possible vortices in small balls of centers $\avj$ and radii $ \varrho_j $, $ j=1\ldots J$, and any statement on the vorticity measure $\mu$ translates into one on the measure 
	\bdm
		2\pi |\log \ep| ^{-1} \sum_{j=1} ^J d_j \delta_{\avj},
	\edm
 thanks to the so-called Jacobian estimate (see Proposition \ref{pro:jacest} below). Note however that the norm in which the estimates hold in this case is necessarily weaker since Dirac masses are less regular than $\mu$.
\end{rem}

We will comment further on these results below. Section \ref{sec:heuristics} presents a sketch of our proofs and Section \ref{sec:discussion} contains a comparison to earlier results and discusses a possible extension of our method to a different setting, namely that of the third critical speed in a flat trap, studied before in \cite{CRY,R}.

\subsection{Sketch of Proofs}
\label{sec:heuristics}

For the convenience of the reader we now sketch the main ideas of our proofs. Several standard techniques will be employed and some of the ideas we use originate in \cite{ABM} and \cite{R} (one can in particular compare the following sketch to \cite[Section 1.2]{R}). We focus on the way the energy lower bound is obtained because we believe this is more helpful in explaining the origin of the renormalized energy \eqref{ren energy}. We also slightly deviate from the actual proof procedure in several places when this serves the purpose of our heuristic considerations.

As is standard, the proof starts with an energy decoupling:
\begin{equation}
\label{eq:sk decouple}
\gpe = \hgpe + \int_{\R ^2} \diff \rv \lf\{ \half g ^2 |\nabla u| ^2 -g ^2 \Om \rvp \cdot (iu,\nabla u) + \frac{g ^4}{\ep ^2} (1-|u| ^2) ^2 \ri\},
\end{equation}
where $(iu,\nabla u ) = \mathbf j$ is the superfluid current defined in \eqref{eq:intro j}. The only input in \eqref{eq:sk decouple} is the variational equation for $g$. Then, as already mentioned, the bulk of the mass is contained in $\tfd$, so that we make a very small error by restricting the integration to this domain. Also, since $g^2 \approx \tfm$, we can, at least at the level of heuristics, consider the reduced functional
\begin{equation}
\label{eq:sk reduc ener} 
\E [u] = \int_{\tfd} \diff \rv \bigg\{ \half \tfm |\nabla u| ^2 -\tfm  \Om \rvp \cdot (iu,\nabla u) + \frac{\tfm ^2}{\ep ^2} (1-|u| ^2) ^2 \bigg\}
\end{equation}
and we essentially have to understand why its minimization reduces to that of $\tfIc$, once a scaling factor $ \OO(|\log\eps|^2) $ has been extracted. It is fairly easy to obtain from the rough upper bound (trial state $v\equiv 1$)
\[
 \E[u] \leq 0,
\]
that the sum of the first and the last term in \eqref{eq:sk reduc ener} are suitably bounded above. Using the vortex balls method introduced independently in \cite{J,Sa}, this allows to control the area of the set where $|u|$ differs significantly from $1$ and enclose it in a finite collection of disjoint small balls $\B_j $, $ j=1,\ldots,J$ that serve as `approximate vortices'. One also proves that each ball contains a kinetic energy whose leading order is precisely the energetic cost we have been alluding to in \eqref{eq:intro vcost}:
\begin{equation}\label{eq:sk vballs}
\int_{B_j} \diff \rv \: \half \tfm |\nabla u| ^2 \gtrapprox  \pi |d_j|   \tfm (a_j) \left| \log \ep \right|,
\end{equation}
where $d_j$ is the degree of $u$ on $\dd \B_j$ and $\avj$ the center of $ \B_j$. 

To evaluate the energetic gain of vortices it is convenient to integrate by parts the second term of \eqref{eq:sk reduc ener} using the potential defined in \eqref{eq:intro pot}:
\begin{equation}\label{eq:sk ipp}
 -\Om\int_{\tfd} \diff \rv \: \tfm  \rvp \cdot (iu,\nabla u) = |\log \ep| \int_{\tfd} \diff \rv \: \tfpot \curl(iu,\nabla u).
\end{equation}
Then according to \eqref{eq:intro mu} and \eqref{eq:intro vortic} (which can be put on a rigorous basis thanks to the Jacobian estimate method \cite{JS2})
\begin{equation}\label{eq:sk js}
|\log \ep| \int_{\tfd} \diff \rv \: \tfpot \curl(iu,\nabla u) \approx 2\pi |\log \ep | \sum_j d_j \tfpot(a_j).   
\end{equation}

Note that the expressions for the energetic cost and gain of vortices can be motivated by a simple computation involving an ansatz of the form 
\[                                                                                                                                                
u (z)= \xi (z) \prod_{j=1} ^J \left(\frac{z-z_j}{|z-z_j|}\right) ^{d_j}                                                                                                                                     \]
where we have used complex notation $ z : = x + i y $ for a point $ \rv = (x,y) \in \R^2 $ and $z_j$ are the positions of the vortices. The real-valued function $\xi$ is a cut-off ensuring that $u$ vanishes close to the vortices. What actually comes out of such a computation is a factor $d_j ^2$ instead of $|d_j|$ in \eqref{eq:sk vballs} but rigorous analysis has so far been limited to the obtention of the smaller $|d_j|$ factor (recall that $d_j$ is an integer), except when a priori information is available on the vortex distribution. The difference is anyway of no concern to us because it is always favorable for the vortices to be singly-quantized, i.e., $d_j=1$.

\medskip

There now remains to bound from below the part of the kinetic energy contained outside the vortex balls. To this end we note that, since outside $\cup_j \B_j$ we have $|u|\approx 1$,  
\[
 \int_{\tfd \setminus \cup_j \B_j} \diff \rv \: \half \tfm |\nabla u| ^2 \approx \int_{\tfd \setminus \cup_j \B_j} \diff \rv \: \half \tfm |\mathbf j| ^2 
\]
with $\mathbf j$ defined in \eqref{eq:intro j}. Now comes the procedure of regularization of the vorticity measure we have been alluding to in Remark \ref{sec:results}.\ref{rem:ren energy}: since the vortex balls contain vortices, i.e., phase singularities, the gradient of the phase of $u$ is expected not to be very well behaved inside them. We thus exclude these regions by setting
\begin{equation}\label{eq:sk jt}
 \jt := \begin{cases}
         \mathbf j, 	&	\mbox{in } \tfd \setminus \cup_j \B_j \\
         0,		&	\mbox{in } \cup_j \B_j,
        \end{cases}
\end{equation}
and write 
\[
 \int_{\tfd \setminus \cup_j \B_j} \diff \rv \: \half \tfm |\nabla u| ^2 \approx \int_{\tfd} \diff \rv \: \half \tfm |\jt| ^2. 
\]
With $\jt$ is associated the vorticity measure (rescaled as $\mu$)
\[
 \mut:= |\log \ep| ^{-1} \curl \: \jt
\]
and in turn we can associate with $\mut$ the potential $h_{\mut}$ defined as in \eqref{eq ren energy}. We then have by definition
\begin{equation}\label{eq:sk j hmut}
\curl \left( \jt + |\log \ep | \frac{1}{\tfm} \nablap h_{\mut}\right) = 0,
\end{equation}
which implies that $\jt$ and $|\log \ep | \frac{1}{\tfm} \nablap h_{\mut}$ differ by the gradient of an $H^1$ function. Then we conclude (see Lemma \ref{lem:low int estimate} for details)
\begin{equation}\label{eq:sk int ener}
 \int_{\tfd \setminus \cup_j \B_j} \diff \rv \: \half \tfm |\nabla u| ^2 \gtrapprox \half |\log \ep| ^2 \int_{\tfd} \diff \rv \: \frac{1}{\tfm} |\nabla h_{\mut}| ^2
\end{equation}
and this term has to be interpreted as the energy due to the interaction between vortices via the potential $h_{\mut}$ they create.

Now, as the term `regularization' suggests, one can actually prove that $\mut \approx \mu$ in a suitable sense akin to that in which \eqref{eq:intro vortic} can be made rigorous. Gathering \eqref{eq:sk vballs}, \eqref{eq:sk ipp}, \eqref{eq:sk js} and \eqref{eq:sk int ener} and dropping the last term in \eqref{eq:sk reduc ener} we may thus write (neglecting all remainder terms for simplicity)
\begin{equation}\label{eq:sk low bound}
\E [u] \gtrapprox |\log \ep| ^2 \int_{\tfd} \lf\{ \frac{1}{2\tfm} |\nabla h_{\mut}| ^2 + \half \tfm |\mut| + \tfpot \mut \right\},
\end{equation}
where we have again used the informal relation \eqref{eq:intro vortic}.
This is of course the desired lower bound, since one can recognize in the right-hand side our renormalized energy functional. 

Now come the main technical novelties of the paper. Observe first that we do not prove \emph{a priori} as is often done \cite{AAB,ABM,IM2,R} that most vortices carry a positive degree. This fact follows directly from the minimization of $\tfIc$. In other words, the minimizer of $\tfIc$ is automatically positive and there is no need to restrict the minimization to positive measures. Indeed, although the problem of minimizing $\tfIc$ bears some resemblance with the obstacle problem in GL theory \cite[Chapter 7]{SS2}, it is in fact a bit simpler and its explicit unique minimizer is given by \eqref{musta}. Contrary to the GL case we do not obtain a free boundary problem: the region where vortices should lie is directly encoded in the cost function \eqref{eq:into costf}.

Moreover, $\tfIc$ has the following very nice stability property (see Section \ref{sec:limit problem} for the proof): for any measure $\nu$  
\begin{equation}\label{eq:sk stability}
 \tfIc [\nu] \geq \tfI + \int_{\tfd} \frac{1}{2 \tfm} |\nabla h_{\nu - \musta}| ^2.
\end{equation}
We can thus conclude from the above that
\[
 \gpe \gtrapprox \hgpe + \tfI |\log \ep| ^2 + |\log \ep| ^2 \int_{\tfd} \frac{1}{2 \tfm} |\nabla h_{\mut - \musta}| ^2
\]
and the last term is a norm squared of $\mut$, corresponding to $\Vert\: .\Vert\:_{\tfm} ^2$ but with the $\| \nabla \phi \|_{L^{\infty}}$ term dropped in \eqref{eq:norm}, as one case easily see by a simple duality argument. 

What remains to be done is the construction of a trial state giving an energy upper bound confirming that the above lower bound is optimal. We adapt a well-established technique (see, e.g., \cite{AAB,ABM,SS2,R}) based on a Green representation of the potentials defined as in \eqref{eq ren energy} to out setting. We use a Riemann  approximation of $\musta$ by a measure located in small balls that mimic vortices and the suitable definition of a phase factor whose $\curl$ is this approximation of $\musta$. We refer to Section \ref{sec:upper bound} for details and simply state the result, again without making the remainder terms precise 
\begin{equation}\label{eq:sk final}
 \hgpe + \tfI |\log \ep| ^2  \gtrapprox \gpe \gtrapprox \hgpe + \tfI |\log \ep| ^2  + |\log \ep| ^2 \int_{\tfd}\frac{1}{2 \tfm}  |\nabla h_{\mut - \musta}| ^2.
\end{equation}
Now an important advantage of our approach becomes apparent, since with such an estimate there is no need to rely on compactness arguments to prove vorticity asymptotics. We obtain from  (the rigorous version of) \eqref{eq:sk final} an estimate of the norm of the difference $\mut-\musta$ and there only remains to use the fact that only a small part of the current is removed in the definition \eqref{eq:sk jt} to estimate the difference between $\mu$ and $\mut$ (see Lemma \ref{lem:regul vortic}) and thus conclude the proof of Theorem \ref{teo:vorticity}. Note again that, since we expect that $\mu$ contains singularities approaching Dirac masses, the regularization procedure leading to $\mut$ seems unavoidable in order to properly define the renormalized energy. This justifies the claim in Remark \ref{sec:results}.\ref{rem:norm} that $\mut-\musta$ can be estimated in a better norm than $\mu-\musta$.

To our knowledge this is the first time in the literature that explicit estimates on the rate of convergence of a vorticity measure to the minimizer of a renormalized energy are provided. There is the exception of \cite{R} but the limit problem was simpler there and its properties not exploited fully as we do here. Common to both papers however is the regularization procedure of the vorticity measure that allows to prove explicit estimates without using any compactness argument.

\begin{rem}(Boundary conditions in \eqref{eq ren energy})\label{rem:dirichlet bc} \mbox{}\\
The reader might wonder why it is natural to use Dirichlet boundary conditions in \eqref{eq ren energy}, i.e., set $ h_{\nu} = 0 $ on $ \partial \B(\rtf) $. A first answer is that this is essential to be able to go from \eqref{eq:sk j hmut} to \eqref{eq:sk int ener} (see the proof of Lemma \ref{lem:regul vortic}). A more physical answer can be provided however: $\mathbf j$ and $\jt$ are superfluid currents, so they must be thought of as velocity fields. In fact they correspond to phase gradients: where one can write $u = |u| e^{i\varphi}$, then $\mathbf j=|u| ^2 \nabla \varphi$ and one should remember that essentially $|u| \approx 1$ except in the small region covered by vortex balls. The potential $h_{\mut}$ is therefore defined modulo a constant and the Dirichlet boundary condition should rather be thought of as the constraint that $h_{\mut}$ is constant on $\dd \B(\rtf) $, i.e., $\dd_{\tau} h_{\mut}=0$ on $\dd \B(\rtf) $, where $\tau$ is the tangent vector. Since $\jt \approx \mathbf j$ should be thought of as $ 1/\tfm\nabla h_{\mut}$ rotated by $\pi/2$ (see \eqref{eq:sk j hmut}), this means that we are actually imposing $\dd_{n} \jt = 0$ on $\dd \B(\rtf) $, a rather natural condition if we think of $ \B(\rtf) $ as the support of the condensate from which the fluid should not escape: the superfluid current must be tangent to the boundary of the sample. 
\end{rem}

\subsection{Discussion and Extensions}\label{sec:discussion}

In many respects, the regime we study here for the GP theory is the analogue of the `obstacle problem regime' of GL theory studied in \cite{SS1} and \cite[Chapter 7]{SS2} (see also \cite{JS1}). Common to both settings is the fact that vortices occupy a region whose size grows when increasing $\Omega$ (respectively the external magnetic field $h_{\rm ex}$ in GL theory) until it fills the whole sample. The limit problem we obtain is related to that of GL theory but has significant differences, mostly due to the inhomogeneous matter density profile of GP theory. It is in some sense simpler because it has an explicit solution, but it leads to richer physics: the inhomogeneity of the vortex distribution and its progressive homogenization we have discussed before are indeed absent in GL theory. 

An analogue of the stability estimate \eqref{eq:sk stability} also holds for the limit problem of GL theory, although it does not seem to have been noticed before (see Remark \ref{sec:limit problem}.\ref{rem:GL} below). It can be used in the manner we do in this paper to obtain slight improvements of the classical results of \cite[Chapter 7]{SS2}, like convergence of the vorticity measure in better norms and explicit error estimates. 

\medskip

Our method can also be adapted to treat variations of the physical setting: for instance it is a rather simple adaptation to prove the results corresponding to our Theorems \ref{teo:gse asympt} and \ref{teo:vorticity} in the case of an annular condensate. Modulo slight modifications one can thus treat the setting of \cite{AAB} in the regime where the rotation is above the critical speed for vortex nucleation but of the same order of magnitude. 

Heavier modifications are on the other hand needed in order to extend the results to a regime close to the third critical speed\footnote{For a description of the physics of the three critical speeds of GP theory, the reader may want to refer to \cite{CPRY1,CPRY2}.} in a flat traps: The GP functional in this case reads\footnote{Note the slightly different units as in \cite{CRY,CPRY1}, i.e., mass equal to $ 1/2 $ and angular velocity $ 2 \Omega $.}
\begin{equation}\label{eq:flat trap}
\gpf[\Psi] =  \int_{\B(1)} \diff \rv \: \bigg\{ \lf| \nabla \Psi \ri|^2 - 2\Omega \Psi^* L \Psi + \frac{\lf| \Psi \ri|^4}{\eps^2} \bigg\},
\end{equation}
where $\B(1)$ is the unit disc, the minimization is performed under a unit mass constraint (that one can supplement with a Dirichlet boundary condition) and the regime of interest is $\Omega \propto \ep ^{-2} |\log \ep| ^{-1}$ with $\ep \ll 1$. We refer to \cite{CRY,R,CPRY1} for a more thorough discussion of this model, but only mention that if
\[
 \Omega = \frac{\Omega_0}{\ep ^2 |\log \ep|}, \qquad \Omega_0 >\frac{2}{3\pi},
\]
one reaches a giant vortex phase where the mass is confined to a thin annulus along the boundary $\dd \B$ of the trap by centrifugal forces and no vortex is present in the annulus. If 
\[
 \Omega_0 < \frac{2}{3\pi},
\]
the present analysis applies and yields the conclusion that vortices are densely packed in an annulus included in the bulk of the condensate. This annulus progressively fills the bulk in the limit $\Omega_0 \to 0$ (that is when $\Omega$ \emph{decreases}) while the vortex density, highly inhomogeneous due to a non-constant density profile when $\Omega_0$ is not too small, gradually homogenizes. The precise expression of the vortex density as a function of the matter density is exactly analogue to \eqref{eq:intro result}, with the appropriate density $\tfm$ and cost function $\tfpot$ (see \cite{CRY}). This bridges between \cite{CRY} and \cite{CY,CPRY1} where we have proved that the vortex distribution is homogeneous when $|\log \ep| \ll \Om \ll \ep^{-2}|\log \ep| ^{-1}$. The borderline case of $\Omega_0 = 2 (3 \pi) ^{-1}(1-o(1))$ is considered in \cite{R}.

As discussed in the introduction of \cite{CRY}, the analysis of the third critical speed in a flat trap is more involved than that of the first critical speed treated here. Our method thus needs to be supplemented with the tools developed in \cite{CRY,R} to adapt to this setting. What makes the application possible is that the phase transition happening at the third critical speed can be seen as a $\Ofirst$ type transition but \emph{backwards} (vortices disappear when $\Om$ is increased). 

Let us also emphasize that the regime corresponding to the third critical speed in a `soft' trap given by a potential such as \eqref{ext pot} is quite different from that in a flat trap (see \cite{CPRY2,CPRY3}) and shares much less features with the $\Ofirst$ regime. It is thus unlikely that the methods we develop here can apply in this case.

\medskip

During the completion of this paper we learned of the recent work \cite{BJOS2} where (among other things) the regime we are dealing with has been studied for a 3D condensate. A limit problem is derived, formulated in terms of the current \eqref{eq:intro j} instead of the vorticity \eqref{eq:intro mu}. The situation seems more complicated in 3D, and nothing as explicit as formula \eqref{eq:intro result} appears to be derivable from this limit problem. The fact that, when $\Omega \gg |\log \ep|$, the vortex distribution becomes homogeneous (in fact constituted of many densely packed and uniformly distributed straight vortex lines parallel to the axis of rotation) is proved however, confirming our results in the regime where the rotation speed largely exceeds the critical one. Whether our approach, supplemented with the tools of \cite{BJOS1}, can be generalized to three space dimensions and complete the results of \cite{BJOS2} remains a question for future investigations. 

\medskip

We finally remark that vortex patterns inhomogeneities play a crucial role in a very different regime than that under consideration here, namely in the vicinity of the maximum rotation speed attainable in a condensate confined by a purely quadratic trap \cite{AB,ABD}. In that case the inhomogeneity only manifests itself close to the boundary of the condensate and is thus hardly observable,  but it has the important effect of modifying the density profile of the fluid.

\subsection{Plan of the Paper}

The rest of the paper presents the proofs of our main results. We start by analyzing in Section \ref{sec:limit problem} the properties of the renormalized energy. 
Section \ref{sec:upper bound} contains the construction of our trial state and the evaluation of its energy. The general technique is not new but, since our method allows to deduce from energy bounds a quantitative estimate of the rate of convergence of $\mu$ to $\musta$, we make an effort to obtain precise estimates of the remainders in the energy upper bound. Section \ref{sec:lower bound} is then concerned with the energy lower bound and the proof of Theorem \ref{teo:vorticity}. An appendix gathers technical results used in several places of the proofs.

\section{The Renormalized Energy}
\label{sec:limit problem}

In this section we focus on the study of the renormalized energy of vortices. For further convenience we define the limit functional for a slightly larger class of densities $ \rho $ and measures $\nu$ as
\begin{equation}\label{eq:defiInu}
\irho [\nu] = \int_{\D} \left\{ \frac{1}{2\rho} |\nabla h_{\nu}| ^2 + \frac{1}{2} \rho |\nu| + F \nu \right\},  
\end{equation}
where 
\begin{equation}\label{eq:defihnu}
	 \begin{cases}
		-\nabla \left( \frac{1}{\rho} \nabla h_{\nu}\right) = \nu,	&	\mbox{in } \D, \\
		h_{\nu} = 0,								&	\mbox{on } \partial \D.
	\end{cases} 
\end{equation}
We will not strive for the most general assumptions allowing the study of such an energy functional, but instead state a theorem that is sufficient for the purpose of proving our main results. From a mathematical point of view, this means that the functions $\rho$ and $F$ and the domain $\D$ appearing in the definition above are not necessarily the ones that appear in the preceding sections, but the assumptions we impose on them are inherited from the physical features of our original GP theory. Note however that we do not assume in this section that the problem is radial. We do not stress the dependence of $\irho$ on $F$ because in the applications we have in mind $F$ is related to $\rho$ as in \eqref{eq:intro pot}.

The following theorem states the existence and uniqueness of the minimizer of the functional $\irho$ in the natural energy space. These are somewhat classical results reminiscent of potential theory \cite{ST,St}, more important to us are the explicit formula \eqref{eq:defi musta} for the minimizer and the stability property \eqref{eq:stability I} that is the key input in the proof of Theorem \ref{teo:vorticity}.

\begin{teo}[\textbf{Minimization of the renormalized energy}]\mbox{}\label{theo:renorm energy}\\ 
Let $\D$ be a regular open subset of $\R ^2$, $\rho \in C^2 (\D),\: \rho \geq 0$, $F \in C^2(\D) $. We assume that 
\begin{equation}\label{eq:hypo signe}
\half \rho - F > 0,	\quad \mbox{in } \D, 
\end{equation}
and define  
\begin{equation}\label{eq:defi H}
\hrho:= \half\rho  + F.
\end{equation}
\begin{enumerate}
\item \textbf{Existence and uniqueness of a minimizer}. The functional $\irho$ has a unique minimizer $\mustar$ in the class of measures 
\beq
	\label{measure class}
	\M_{\rho} (\D) = \left\lbrace \nu \in \left(C^0_c (\D)\right) ^*,\: \int_{\D} \lf\{ \frac{1}{\rho} |\nabla h_{\nu}| ^2 +  \rho |\nu| \ri\} < +\infty \right\rbrace.                                                                                                                                        \eeq
It is given by the formula
\begin{equation}\label{eq:defi musta}
\mustar = \left[\nabla \lf( \frac{1}{\rho} \nabla \hrho \ri) \right]_+ \one_{\left\lbrace  \hrho \leq 0\right\rbrace},
\end{equation}
and the ground state energy is
\beq
	\label{gs ren energy}
	\irhoe= \irho[\mustar] =  \frac{1}{2} \int_{\supp (\mustar)} \hrho \mustar.
\eeq
\item \textbf{Stability of the minimizer}. For any $\nu \in \M_{\rho} (\D)$
\begin{equation}\label{eq:stability I}
\irho[\nu] \geq \irhoe + \int_{\D} \frac{1}{2\rho} \left| \nabla h_{\mustar - \nu} \right| ^2. 
\end{equation}
\end{enumerate} 
\end{teo}

The main physically relevant assumption we make here is \eqref{eq:hypo signe}, which indicates that negative degree vortices are not energetically favored. In our setting it will always be satisfied because typically $\rho \geq 0$ and $F\leq 0$ with equality only at the boundary of the domain. One may certainly prove a related theorem in the case where \eqref{eq:hypo signe} does not hold but it is not our concern here.

\begin{proof}
We split the proof in four steps.

\emph{Step 1 (Existence and Euler-Lagrange equation).} It is not difficult to prove the existence of a minimizer, we will thus skip most of this discussion. Let us just note that with our assumptions
\[ 
\irho[\nu] \geq \int_{\D} \frac{1}{4\rho} |\nabla h_{\nu}| ^2 + C_1 \left\Vert \nu \right\Vert_{H^{-1}(\D)} ^2 - C_2 \left\Vert \nu \right\Vert_{\left( C^1_c(\D)\right) ^*} + \int_{\D} \rho |\nu|,
\]
with $ C_1 = (2\sup_{\rv \in \D} \rho)^{-1} $ and $ C_2 = \sup_{\rv \in \D} (\half \rho - F) $. Recalling the embedding of $\left( C^1_c(\D)\right) ^*$ in $H^{-1}(\D)$ it is then easy to deduce bounds on the minimizing sequences and conclude by lower semi-continuity arguments.\\
Considering now a variation of the form $(1+tf)\mustar $, $ f\in C^0 (\D)$, and noticing that \eqref{eq:defihnu} implies
\begin{equation}\label{eq:reform inter}
\int_{\D} \frac{1}{\rho} |\nabla h_{\nu}| ^2 = \int_{\D} h_{\nu} \nu,
\end{equation}
we see that the Euler-Lagrange equation of the minimization problem takes the form
\begin{equation}\label{eq:EEL vorticity}
\int_{\D} \left( \hstar \mustar + \half \rho|\mustar| + F \mustar \right) f = 0
\end{equation}
for any $f\in C^0 (\D)$. Writing 
\[
 \mustar = \mustar^+ - \mustar^-, \quad \mbox{with } \mustar ^+,\: \mustar^- \geq 0, 
\]
we deduce from the above that 
\begin{equation}\label{eq:prop hstar}
\begin{cases}
\hstar = - \hrho,	& \mbox{on } \supp \left(\mustar ^+ \right),\\
\hstar = \half \rho  - F,  & \mbox{on } \supp \left(\mustar ^- \right),
\end{cases}
\qquad
\begin{cases}
\mustar = \nabla \lf( \frac{1}{\rho} \nabla \hrho \ri),	& \mbox{on }\supp \left(\mustar ^+ \right),\\
\mustar = -\nabla \lf[ \frac{1}{\rho} \nabla \left(\half \rho - F \right)\ri],	& \mbox{on }\supp \left(\mustar ^- \right).
\end{cases}
\end{equation}
Note that this implies in particular $\supp \left(\mustar ^+ \right) \subset \left\lbrace \nabla \lf(\frac{1}{\rho} \nabla \hrho\ri) \geq 0 \right\rbrace$.

\emph{Step 2 (Explicit formula for the minimizer).} We first prove that $\mustar^- = 0$ by using the trial state 
\[
 \mutrial = \mustar ^+.
\]
Uniqueness of the solution to \eqref{eq:defihnu} implies that
\[
h_{\mutrial} =  - \hrho,	\quad \mbox{on } \supp \left(\mustar ^+ \right),
\]
$h_{\mutrial}$ being simply extended to the whole domain $\D$ by requiring that $\nabla (\rho ^{-1} \nabla h_{\mutrial}) = 0$ on $\D \setminus \supp \left(\mustar ^+ \right)$. Recalling \eqref{eq:reform inter} and \eqref{eq:prop hstar}, we thus have
\[
\irho[\mutrial] = \frac{1}{2} \int_{\D} \mustar ^+ \hrho,
\]
whereas 
\[
\irho[\mustar] =   \frac{1}{2}  \int_{\D} \mustar ^+ \hrho +  \frac{1}{2}  \int_{\D} \left(\half \rho - F\right) \mustar^-
\]
Using assumption \eqref{eq:hypo signe}, this clearly means that it must be $\mustar^- = 0$.
The minimizer $\mustar\geq 0$ is then completely determined by the set $\supp \left(\mustar \right)\subset \{ \nabla (\frac{1}{\rho} \nabla \hrho) \geq 0 \}$ and we have (recall \eqref{eq:prop hstar}) 
\[
\irho[\mustar] =  \frac{1}{2} \int_{\supp (\mustar)} \hrho \left(\nabla \frac{1}{\rho} \nabla \hrho \right),
\]
from which it is easy to deduce that, in order for $\mustar$ to minimize $\irho$, it must be
\[
\supp \left(\mustar \right)= \left\lbrace \nabla \lf( \frac{1}{\rho} \nabla \hrho \ri) \geq 0 \right\rbrace \cap \left\lbrace  \hrho \leq 0 \right\rbrace,                                                                                         
\]
and \eqref{eq:defi musta} follows.

\emph{Step 3 (Consequences).} We first claim that 
\begin{equation}\label{eq:techn hstar 1}
\hstar \geq -\hrho, \quad \mbox{a.e. in } \D.
\end{equation}
Indeed, we have equality in $\supp(\mustar)=\{ \nabla (\frac{1}{\rho} \nabla \hrho) \geq 0 \} \cap \left\lbrace  \hrho \leq 0 \right\rbrace$. Let us now suppose that \eqref{eq:techn hstar 1} does not hold in some subregion of $\supp(\mustar)^c$ and pick some $\nu>0$ with support in this region. Consider the trial state 
\[
 \mutrial = \mustar + t \nu 
\] 
for $t>0$. Since $\mustar$ and $\nu$ have disjoint supports 
\[
|\mutrial| = | \mustar| + t \nu. 
\]
Also we remark that for any pair of measures $\mu_1,\mu_2 \in \M (\D)$
\begin{equation}\label{eq:two meas interact}
\int_{\D} \frac{1}{\rho} \nabla h_{\mu_1} \cdot \nabla h_{\mu_2} = \int_{\D} \mu_1 h_{\mu_2},  
\end{equation}
which allows to compute the energy of $\mutrial$ and find
\[
\irho[\mutrial] = \irho[\mustar] + t ^2 \int_{\D}\frac{1}{2\rho}|\nabla h_{\nu}| ^2 + t\int_{\D}\left(\hrho + \hstar \right)\nu.
\]
Since the last term is negative by assumption, for $t$ small enough we would obtain $\irho[\mutrial] < \irho[\mustar]$, which contradicts the fact that $\mustar$ minimizes $\irho$.

Next we prove that 
\begin{equation}\label{eq:techn hstar 2}
\half \rho - F -  \hstar \geq 0, \quad \mbox{a.e. in } \supp(\mustar).
\end{equation}
The argument is again by contradiction. Suppose that in some subregion of $\supp(\mustar)$, we have $ \rho - 2 F - 2 \hstar <0$ and pick some positive measure $\nu \in \M (\D)$ whose support is included in this region. Consider the trial state 
\[
 \mutrial = \mustar - t \nu 
\]
for some $t>0$. We have
\[
 \left| \mustar - t \nu\right| \leq |\mustar| + t \nu.
\]
because of \eqref{eq:defihnu}. Computing $\irho[\mutrial]$ by expanding the quadratic term and using \eqref{eq:two meas interact}, we thus find
\[
\irho[\mutrial] \leq \irho[\mustar] + t ^2 \int_{\D}\frac{1}{2\rho} |\nabla h_{\nu}| ^2 + t \int_{\D} \left( \half \rho -F -  \hstar \right) \nu,
\]
from which we immediately deduce that, for small enough $t$, we would obtain $\irho[\mutrial]<\irho[\mustar]$, which is a contradiction. We conclude that \eqref{eq:techn hstar 2} must hold true.

\emph{Step 4 (Stability).} We are now ready to prove the most important point of our theorem, namely the stability property \eqref{eq:stability I}. We write any $\nu \in \M(\D)$ as 
\[
 \nu = \mustar + \nu_1 + \nu_2,
\]
where $\supp( \nu_1 ) \subset \supp (\mustar)$ and $\supp( \nu_2 )\subset \supp(\mustar) ^c$ and note that 
\[
 |\nu| \geq \mustar + \nu_1 + |\nu_2|.
\]
Since $\mustar \geq 0$, using \eqref{eq:two meas interact} again, we get 
\begin{equation}\label{eq: proof stability}
\irho[\nu] \geq \irho[\mustar] + \int_{\D} \frac{1}{2\rho} |\nabla h_{\nu_1+\nu_2}| ^2 + \int_{\D} \left( \half \rho  + F + \hstar\right) \nu_1 + \int_{\D} \half \rho  |\nu_2| + \int_{\D} \left( F + \hstar\right) \nu_2,
\end{equation}
but $\rho  + 2 F + 2 \hstar = 0$ on $\supp (\nu_1)$ by \eqref{eq:prop hstar}. Thus, decoupling $ \nu_2 $ into its positive and negative parts:
\[
 \nu_2 = \nu_2 ^+ - \nu_2 ^-, \quad \mbox{with }  \nu_2 ^+,\: \nu_2 ^- \geq 0,
\]
we have
\[
\irho[\nu] \geq \irho[\mustar] + \int_{\D} \frac{1}{2\rho} |\nabla h_{\nu_1+\nu_2}| ^2  + \int_{\D} \left(\half \rho + F +  \hstar\right) \nu_2 ^+ +\int_{\D} \left( \half \rho - F -  \hstar\right) \nu_2 ^-.
\]
The last two terms of the above expression are positive because of \eqref{eq:techn hstar 1} and \eqref{eq:techn hstar 2} and this leads to \eqref{eq:stability I}.
\end{proof}

\begin{rem}(The obstacle problem of Ginzburg-Landau theory)\label{rem:GL}\mbox{}\\
As already mentioned, the problem studied in this section has a clear connection with the limit problem obtained in the study of type-II superconductors in the first critical field regime \cite{SS1,SS2}, which in turn is connected to obstacle problems \cite{BS}. The energy of a vortex density in a type-II superconductor occupying the 2D domain $\D$ can be approximated in suitable units and variables by 
\begin{equation}\label{eq:GL ener}
 \I^{\rm GL} [\mu] = \frac{1}{2} \int_{\D} \left(\frac{|\mu|}{\lambda}  + \left|\nabla h_{\mu} \right| ^2 + |h_{\mu}-1| ^2 \right),
\end{equation}
where $h_{\mu}$ solves
\begin{equation}\label{eq:GL pot}
\begin{cases}
 -\Delta h_{\mu} + h_{\mu} = \mu,	&	\mbox{in } \D, \\
 h_{\mu} = 1,				&	\mbox{on } \dd \D,
\end{cases} 
\end{equation}
and $\lambda>0$ is a parameter.

It is known \cite{BS,SS1} that \eqref{eq:GL ener} has a unique minimizer $\mu^{\rm GL}$, expressed in terms of the solution of an obstacle problem, which is a particular type of free-boundary problem. As far as we know, no explicit formula for $\mu^{\rm GL}$ exists for generic domains $ \D $: it is constant in a subdomain of the sample $\D$ and zero in the rest of the domain but the boundary between the two regions is not known explicitly, although much can be proved about it (see \cite[Chapter 7]{SS2} and references therein). In the special case of a spherically symmetric domain, one can show by uniqueness of the minimizer that $ \mu^{\rm GL} $ is in fact radial and therefore almost explicit.

Using a method similar to that we used for the proof of Theorem \ref{theo:renorm energy}, one can show the following: for any $\nu$ such that $\mu^{\rm GL}+\nu$ is in the energy space corresponding to \eqref{eq:GL ener},
\begin{equation}\label{eq:GL stability}
 \I^{\rm GL} \left[\mu^{\rm GL}+\nu\right] \geq \I^{\rm GL} \left[\mu^{\rm GL}\right] + \frac{1}{2} \int_{\D} \left(\left|\nabla h_{\nu} \right| ^2 + |h_{\nu}-1| ^2 \right),
\end{equation}
which is the equivalent of the stability estimate \eqref{eq:stability I}.
A convenient way of seeing that \eqref{eq:GL stability} holds true is the change of variables $\mu \to \mu-1$, $h_{\mu} \to h_{\mu}-1$, which yields as equivalent problem the minimization of the functional
\bdm
 \tilde{\I}[\mu] = \frac{1}{2} \int_{\D} \left(\frac{|\mu+1|}{\lambda}  + \left|\nabla h_{\mu} \right| ^2 + |h_{\mu}| ^2 \right) = \frac{1}{2}\int_{\D} \left(\frac{|\mu+1|}{\lambda}  + h_{\mu} \mu \right),
\edm
where now
\[
 \begin{cases}
  -\Delta h_{\mu} + h_{\mu} = \mu,	&	\mbox{in } \D, \\
 h_{\mu} = 0,				&	\mbox{on } \dd \D.
 \end{cases}
\]
This form has the advantage of being closer to \eqref{eq:defiInu} and facilitating the proof of \eqref{eq:GL stability}, but the physical interpretation of the problem is more transparent in \eqref{eq:GL ener}.
\end{rem}

Now we apply the result proven above to the functional $ \tfIc $ defined in \eqref{ren energy}. In this case both $ F $ and $ \hrho $ can be explicitly computed:
	\beq
		\label{explicit F}
		\tfpot(r) = - \Omega_0 \int_r^{\rtf} \diff t \: t \: \tfm(t) = - \frac{1}{4} \Omega_0 \lf[ \rtf^s \lf( \rtf^2 - r^2 \ri) - \frac{2}{s+2} \lf( \rtf^{s+2} - r^{s+2} \ri) \ri],
	\eeq
	\bml{
		\label{explicit H}
		\tfH(r) = \half \tfm(r) + \tfpot(r)	\\
		 = \frac{1}{4} \lf( \rtf^s - r^s \ri) - \frac{1}{4} \Omega_0 \lf[  \rtf^s \lf( \rtf^2 - r^2 \ri) - \frac{2}{s+2} \lf( \rtf^{s+2} - r^{s+2} \ri) \ri].
	}
	There is a certain freedom in the choice of the integration domain $ \D $ (see Remark \ref{sec:limit problem}.\ref{rem:domain}) in $ \tfIc[\nu] $ but for clarity we set $ \D = \B(\rtf) $. 

\begin{cor}[{\bf Minimization of \mbox{$ \tfIc[\nu] $}}]
	\label{min ren energy: cor}
	\mbox{}	\\
	Let $ \tfIc[\nu] $ be the renormalized energy defined in \eqref{ren energy} with $ \tfm $ given by \eqref{tfm}. 
	There exists a unique minimizer $ \musta $ in the class of measures \eqref{measure class}, which
	is radial and absolutely continuous w.r.t. the Lebesgue measure, i.e., there exists a continuous radial function $ m_{\star}(r) $ such that 
	\begin{equation}
	\musta = \left[m_{\star}(r)\right]_+ \one_{\lf\{ \tfH \leq 0 \ri\}} \: \diff \rv  
	\end{equation}
 and
	\beq
		\label{ms}
		\ms(r) =  \lf[ \half \partial_r ^2 \log \lf( \tfm(r) \ri) + 2\Omega_0 \ri].
	\eeq
	If in addition  $ \Omega_0 > \Omega_1 $, the support of $ \musta $ satisfies
	\beq
		\label{rrs}
		\emptyset \neq \supp(\musta) = \B(\rrs) \subsetneq \B(\rtf),		
	\eeq
	for some $ 0 < \rrs < \rtf $ such that
	\begin{equation}\label{eq:Rstar to Rtf}
	 \rrs \tfr
	\end{equation}
when $\Om_0 \to \infty$.
\end{cor}

\begin{rem}({\it Domain of $ \tfIc[\nu] $})\label{rem:domain}
	\mbox{}	\\
	A simple inspection of the proof of the above results shows that the minimization and thus the ground state energy and minimizer are to some extent independent of the domain $ \D $ in the definition of $ \tfIc[\nu] $. More precisely one obtains the same minimizing $ \musta $ and the same minimum energy for any $ \D $ such that
	\beq
		\label{int domain}
		\B(\rrs) \subset \D \subset \B(\rtf).
	\eeq
	Indeed $ \musta $ is a minimizer as long as $ \supp(\musta) \subset \D $, so that \eqref{rrs} yields the condition \eqref{int domain}.
\end{rem}

\begin{proof}
	In order to apply Theorem \ref{theo:renorm energy}, we can take any $ \D \subset \B(\rtf) $, since by construction $ F(r) < 0 $ for any $ r < \rtf $ and $ \tfpot(\rtf) = 0 $. Moreover $ \tfm $ and $ \tfpot $ clearly satisfy all the requirements of Theorem \ref{theo:renorm energy} and the first part of the statement is thus proven.

	The explicit expression of $ \musta $ is provided by \eqref{eq:defi musta} and since both $ \tfm $ and $ \tfH $ are radial, $ \musta $ must be so as well. A simple computation yields $ \ms $, which is a well-defined continuous function for any $ r < \rtf $. The condition $ \tfH \leq 0 $ guarantees that the support of $ \musta $ is strictly contained in $ \B(\rtf) $: Indeed one has $ \tfH(\rtf) = 0 $ but
	\bdm
		\tfHp(\rtf) = \half {\tfm}^{\prime}(\rtf) = - \frac{1}{4} s \rtf^{s-1} < 0,
	\edm
	which implies that there exists some $ R_1 < \rtf $, such that $ \tfH(R_1) = 0 $ and $ \tfH(r) > 0 $ for any $ R_1 < r < \rtf $. The fact that $ R_1 > 0 $ follows from the analysis of $ \tfH $ at the origin: using \eqref{tfchem} and \eqref{eq:first critical speed}, we compute
	\beq
		\tfH(0) = \frac{1}{4} \tfchem - \frac{1}{2\pi} \Omega_0 =  \frac{1}{2\pi} \lf( \Omega_1 - \Omega_0 \ri) < 0,
	\eeq
	by assumption. Moreover a simple analysis of the derivative of $\tfH $, i.e.,
	\bdm
		\tfHp(r) = \half r \lf( - s r^{s-2} + \Omega_0 \rtf^s - \Omega_0 r^s \ri),
	\edm
	shows that $ \tfHp(r) = 0 $ for $ r = 0, \bar{R} $ for some $\bar{R}>0$. It is positive between $ 0 $ and $ \bar{R} $ and negative for $ \bar{R} < r \leq \rtf $. In particular this, together with the discussion above, implies that $ \tfH $ has a unique maximum in $ [0,\rtf] $ and, since $ \tfH $ is negative at the origin and vanishes at the outer boundary, one can conclude that
	\beq
		\tfH(R_1) = 0,	\qquad	\tfH(r) < 0,	\quad	\mbox{for any } 0 \leq r < R_1.
	\eeq
	To complete the proof it suffices therefore to exploit the explicit expression of $ \ms $, i.e.,
	\beq
		\ms(r) = - \frac{r^{s-2} \lf( r^s + s(s-1) \rtf^s \ri)}{8 \tfm(r)^2} + 2 \Omega_0,
	\eeq
	which easily implies that $ \ms $ is decreasing since the first term on the r.h.s. of the above expression is. Moreover
	\beq
		\ms(0) = \Omega_0 > 0,	\qquad 	\ms(r) \xrightarrow[r \to \rtf]{} - \infty,
	\eeq
	so that the subregion where $ \ms $ is positive is another ball of radius $ R_2 < \rtf $ and 
	\beq
		\rrs = \min[R_1,R_2].
	\eeq 
	To see that $\rrs \to \rtf$ in the limit $\Omega_0 \to \infty$, it is sufficient to note that the maximum point $ \bar{R} $ of $  \tfH $ converges to $ \rtf $ as $ \Omega_0 \to \infty $, which implies that $ \sup_{0 \leq r \leq \rtf} \tfH \to 0 $ in the same limit. Hence
	\beq
		\label{R1 conv}
		R_1 \xrightarrow[\Omega_0 \to \infty]{} \rtf
	\eeq 
	and an analogous statement holds true for $ R_2 $, i.e.,
	\beq
			\label{R2 conv}
		R_2 \xrightarrow[\Omega_0 \to \infty]{} \rtf
	\eeq
	 as it immediately follows by noticing the first terms in  \eqref{ms} is independent of $\Omega_0$, while the second is linear in $ \Omega_0 $.
\end{proof}

A simple computation yields the equations satisfied by $ R_1 $ and $ R_2 $: setting $ r_i : = R_i/\rtf $, $ i = 1,2 $ (so that $ 0 < r_i < 1 $), one has
\beq
	\label{R1 and R2}
	\frac{1- r_1}{1 - r_1^{2/s} - \frac{2}{s+2} \lf( 1 - r_1^{(s+2)/s} \ri)} = \Omega_0 \rtf^2,	\qquad	\half r_2^{(s-2)/s} \frac{s(s-1) + r_2}{(1 - r_2)^2} = \Omega_0 \rtf^2,
\eeq
which allows to obtain a relation between $ r_1 $ and $ r_2 $ independent of $ \Omega_0 $. The computation of $ r_1 $ or $ r_2 $ or their ratio can not be made explicitly but the above equations could be tested numerically to study the dependence of the support of $ \musta $ on the parameter $ s $, that is whether $ \rrs $ is given by $ R_1 $ or $ R_2 $. In the harmonic case $ s = 2 $ however both equations  \eqref{R1 and R2} are solvable and one obtains
\beq
		r_1 = \frac{R_1}{\rtf} = \sqrt{1 -  \frac{2}{\Omega_0 \rtf^2}},	\qquad	r_2 = \frac{R_2}{\rtf} = \sqrt{1 -  \frac{1}{4\Omega_0 \rtf^2} \lf( \sqrt{1 + \frac{3}{2 \Omega_0 \rtf^2}} - 1 \ri)}.
\eeq
Note that, by taking the limit $ \Omega_0 \to \infty $ of the above expressions, one easily recovers \eqref{eq:Rstar to Rtf}.

\section{Energy Upper Bound}
\label{sec:upper bound}

The main result of this section is stated in the following Proposition \ref{upper bound: pro} and is the proof of an appropriate upper bound for the GP ground state energy. 

The result is obtained by testing the GP functional on some explicit trial function $ \trial $. The most difficult part is the construction of such a function because of several requirements it has to fulfill in order that its energy is suitably close to the ground state energy. There are indeed two main constraints: On the one hand the modulus of $ \trial $ must be approximately equal to $ g $, in order to recover the leading order contribution $ \hgpe $, but, on the other hand, its phase has to contain a large number (of order $ \OO(|\log\eps|)  $) of vortices of unit degree, distributed according to the minimizing measure $ \musta $ given by \eqref{eq:defi musta} on a scale $ |\log\eps| $. In addition $ \trial $ must belong to $ \gpdom $, which in particular implies that it must be normalized in $ L^2(\R^2) $. The rest of the proof is just the evaluation of $ \gpf $ on $ \trial $.

\begin{pro}[{\bf GP energy upper bound}]
	\label{upper bound: pro}
	\mbox{}	\\
	If $ \Omega = \Omega_0 |\log\eps| $, with $ \Omega_0 > \Omega_1 $ as $ \eps \to 0 $, then
	\beq
		\gpe \leq \hgpe + \tfI |\log\eps|^2 \left( 1+ \OO\left(\frac{\log|\log\eps|}{|\log \ep| ^{1/2}}\right)\right).
	\eeq
\end{pro}

\subsection{The Trial Function}

	As we have anticipated at the beginning of the section, the most difficult part of the proof is the identification of the proper trial function to use in the estimate of the ground state energy. In order to simplify the analysis however we first extract from $ \gpe $ the leading order term, i.e., the energy $ \hgpe $. Like in the lower bound, this can be done by a splitting trick, i.e., setting $ \gpm =: g v $, where $ g $ is the vortex-free profile given by the minimization \eqref{hgpdom} and $ v $ some locally bounded function. Exploiting the variational equation for $ g $, one obtains the identity
	\beq
		\label{reduced energy}
		\gpe = \hgpe + \E[v],	\qquad	\E[v] : = \int_{\R^2} \diff \rv \: \lf\{ \half g^2 \lf| \nabla v \ri|^2 - \Omega g^2 \rv^{\perp} \cdot \lf(i v, \nabla v\ri) + \eps^{-2} g^4 \lf(1 - |v|^2 \ri)^2 \ri\},
	\eeq
	where we recall that
	\bdm
		\rv^{\perp} : = (-y , x),	\qquad	\lf(i v, \nabla v\ri) : = \frac{i}{2} \lf(v \nabla v^* - v^* \nabla v \ri).
	\edm

	It is clear that in order to obtain an upper bound for $ \gpe $, it is sufficient to provide a suitable trial function $ \vtrial $ to evaluate $ \E[v] $. Note that $ v $ inherits an $L^2$-normalization from $ \gpm $: 
	\bdm
		\int_{\R^2} \diff \rv \: g^2 |v|^2 = 1,
	\edm
	and the same must apply to $ \vtrial $. We thus set
	\beq
		\label{vtrial}
		\vtrial(\rv) : = \ceps \: \xi(\rv) \exp \lf\{ i \phitrial \ri\},
	\eeq
	where $ \xi $ is a cut-off function, $ \phitrial $ is a phase factor we are going to define and $ \ceps $ is a normalization constant. 
		
	The phase factor of $ \vtrial $ must contain a number $ \OO(|\log\eps|) $ of vortices of unit degree distributed according to the measure $ \musta|\log\eps| $. For further convenience we start by defining the phase inside a ball $ \B(\rsm) $ of radius
	\beq
		\label{rsm}
		\rsm : = \rtf - \eps^{2/3} |\log\eps|^{2/3},
	\eeq
	which is strictly contained inside the TF support $ \B(\rtf) $ and has the crucial property that inside $ \B(\rsm) $ the pointwise estimate \eqref{g point est} applies and $ g^2 \approx \tfm $: There we 
	define $ \phitrial $ as the solution of the equation
	\beq
		\label{phi equation 1}
		\nabla \phitrial = - \tfm^{-1} \nablap h_{\nutrial},	\quad	\mbox{for } \rv \in \B(\rsm)\setminus \bigcup_{i=1}^{N_{\eps}} \B(\ai,\eps),	
	\eeq
	where the potential $ h_{\nutrial} $ solves the differential equation (see below for further details on the existence and properties of such a solution)
	\beq
		\label{eq nutrial}
		\begin{cases}
			- \nabla \lf( \tfm^{-1} \nabla h_{\nutrial} \ri) = \nutrial,		&	\mbox{in } \B(\rsm),	\\
			- \Delta h_{\nutrial} = 0,							&	\mbox{in } \B(\rbl) \setminus \B(\rsm),	\\
			h_{\nutrial} = 0,							&	\mbox{on } \partial \B(\rbl),
		\end{cases}
	\eeq
	for some $ \rbl > \rtf $ given below.
	Note that we modify the equation to $ \Delta h_{\nutrial} = 0$ in the boundary layer $\B(\rbl) \setminus \B(\rsm)$. This is to avoid problems due to the vanishing of $\tfm$ on $\dd \B(\rtf)$. Other strategies are possible, see, e.g., \cite[Section 2]{R}.
	The outer radius $ \rbl $ is larger but close enough to $ \rtf $: for technical reasons which will be clearer later and are mostly related to the estimate of the interaction term and the exponential decay \eqref{g exp decay}, we pick
	\beq
		\rbl : = \rtf + \eps^{2/3} |\log\eps|^{4/3}.
	\eeq
	The function $ \nutrial $ is a smooth approximation of a measure given by the sum of Dirac masses at the points $ \ai $, $ i = 1, \ldots, \neps $, e.g.,
	\beq
		\label{nutrial}
		\nutrial : = \frac{2}{\eps^2} \sum_{i = 1}^{\neps} \one_{\B(\ai,\eps)},
	\eeq
	i.e., $ \ai $ are the positions of vortices and $ \neps $ their total number. We have also implicitly assumed that $ \inf |\ai - \mathbf{a}_j| > 2 \eps $ and $ a_i \leq \rsm - \eps $ (see below for further details on the distribution of points). The crucial property of $ \phitrial $ is that, given any contour $ \mathcal{C} \subset \B(\rsm) $ encircling only one ball $ \B(\ai,\eps) $, one has
	\bml{
 		\label{phase pro 1}
		\oint_{\mathcal{C}} \diff \sigma \: \partial_{\tau} \phitrial = - \int_{\partial \B(\ai,\eps)} \diff \sigma \: \partial_{\tau} \phitrial =  \int_{\partial \B(\ai,\eps)} \diff \sigma \: \tfm^{-1} \partial_{n} h_{\nutrial} 	\\
		= - \int_{\B(\ai,\eps)} \diff \rv \: \nabla \lf( \tfm^{-1} \nabla h_{\nutrial} \ri) = \int_{\B(\ai,\eps)}  \nutrial = 2 \pi,
	}
	since the integral of $ \curl (\nabla \phitrial) $ over the area delimited by $ \mathcal{C} $ and $ \partial \B(\ai,\eps) $ vanishes thanks to \eqref{eq nutrial}. Hence $ \phitrial $ is a well-defined phase factor for any $ \rv \in \B(\rsm) \setminus \cup \B(\ai,\eps) $, whereas inside the vortex balls $ \phitrial $ is not the phase of any function and we will have to use a cut-off function $ \xi $ to exclude that region. We need however to define the trial function and therefore its phase up to $ \rbl $ because otherwise the last term in $ \E[v] $ (see \eqref{reduced energy}) would give a too large contribution. For the phase in the boundary layer we set
	\beq
		\label{phi equation 2}
		\nabla \phitrial = - \tfm^{-1}(\rsm) \nablap h_{\nutrial},	\quad		\mbox{for } \rv \in \B(\rbl) \setminus \B(\rsm). 
	\eeq
	In fact since the boundary layer $ \B(\rbl) \setminus \B(\rsm) $ is vortex-free by construction, it is not difficult to realize that the differential equation \eqref{eq nutrial} can be solved explicitly there and
	\bdm
		h_{\nutrial} = c \log\lf(r/\rbl\ri),
	\edm
	where the coefficient $ c $ is fixed by imposing continuity of $ \partial_n h_{\nutrial} $ on $ \partial \B(\rsm) $. On the other hand if we integrate \eqref{eq nutrial} in $ \B(\rsm) $ we obtain
	\beq
		- \int_{\B(\rsm)} \diff \rv \: \nabla \lf( \tfm^{-1} \nabla h_{\nutrial} \ri) = -  \tfm^{-1}(\rsm) \int_{\partial \B(\rsm)} \diff \sigma \:\partial_n h_{\nutrial} = 2\pi \neps,
	\eeq
	which combined with the explicit expression of $ h_{\nutrial} $ for $ r \geq \rsm $ yields
	\beq
		\label{hnutrial outside}
		h_{\nutrial}(r) = - \neps \tfm(\rsm) \log\lf(r/\rbl\ri),	\quad	\mbox{for } \rv \in \B(\rbl) \setminus \B(\rsm),
	\eeq
	and therefore
	\beq
		\label{phi equation 3}
		\nabla \phitrial = \neps \frac{\rv^{\perp}}{r^2},	\quad	\mbox{for } \rv \in \B(\rbl) \setminus \B(\rsm). 
	\eeq
	One can verify that this extension defines a phase in $ \B(\rbl) \setminus \cup \B(\ai,\eps) $: the property \eqref{phase pro 1} is still valid and, for any contour $ \mathcal{C} \subset \B(\rbl) $ containing all the vortex balls $ \B(\ai,\eps) $, one has
	\beq
		\label{phase pro 2}
		\oint_{\mathcal{C}} \diff \sigma \: \partial_{\tau} \phitrial = \int_{\partial \B(\rsm)} \diff \sigma \: \partial_{\tau} \phitrial = 2 \pi \neps,
	\eeq
	thanks to \eqref{phi equation 3} and the fact that the r.h.s. of \eqref{phi equation 1} or \eqref{phi equation 2} is irrotational outside $ \mathcal{C} $.

	We can now turn to the cut-off function $ \xi $. Its role is twofold: on the one hand it is needed to restrict the integration domain in $ \E[v] $ to $ \B(\rbl) $ and, on the other, regularizes $ \phitrial $ inside the vortex balls to ensure that $ \vtrial $ is a one-valued function. Concretely we set
	\beq
		\label{cutoff xi}
		\xi(\rv) : = \xi_{0}(r) \prod_{i = 1}^{\neps} \xivor(|\rv - \ai|),
	\eeq
	with $ \xi_0 $ and $ \xivor $ smooth functions satisfying
	\beq
		\xi_0(r) = 
		\begin{cases}
			0,	&	\mbox{for } r \geq \rbl,	\\
			1,	&	\mbox{for } r \leq \rbl - \eps^{2/3} |\log\eps|^{-1},
		\end{cases}	
		\hspace{1cm}
		\xivor(r) = 
		\begin{cases}
			0,	&	\mbox{for } r \leq \eps,	\\
			1,	&	\mbox{for } r \geq 2\eps,
		\end{cases}
	\eeq 
	and such that their gradients satisfy the conditions
	\beq
		\label{est xi gradients}
		\lf\| \nabla \xi_0 \ri\|_{\infty} \leq \OO(\eps^{-2/3}|\log\eps|^{1/3}),	\hspace{1cm}	 \lf\| \nabla \xivor \ri\|_{\infty} \leq \OO(\eps^{-1}).
	\eeq
	
 	With these definitions we can estimate the normalization constant without saying anything more about the distribution of points but only assuming that $ \neps = \OO(|\log\eps|) $: since the cut-off function satisfies $ \xi \leq 1 $, one clearly has $ \ceps \geq 1 $ and
	\bmln{
		1 = \ceps^2 \int_{\R^2} \diff \rv \: g^2 \xi^2 \geq \ceps^2 \int_{\B(\rbl - \eps^{2/3}|\log\eps|^{-1}) \setminus \cup \B(\ai,2\eps)} \diff \rv \: g^2	\\
		\geq \ceps^2 \bigg[1 - \int_{\R^2 \setminus \B(\rbl - \eps^{2/3}|\log\eps|^{-1})} \diff \rv \: g^2  - \OO(\eps^2|\log\eps|) \bigg] \geq \ceps^2 \lf[ 1 - \OO(\eps^2 |\log\eps|) \ri],
	}
	thanks to boundedness and monotonicity of $ g $ and the exponential decay \eqref{g exp decay}, yielding
	\beq
		g^2(r) = \OO(\eps^{\infty}),	\quad	\mbox{uniformly in } r \geq \rtf + \OO(\eps^{2/3}|\log\eps|^{4/3}).
	\eeq 
	Therefore we conclude
	\beq
		\label{ceps est}
		1 \leq \ceps^2 \leq 1 +  \OO(\eps^{2} |\log\eps|).
	\eeq

	Finally we discuss the distribution of vortex points $ \ai $: the final goal is the reproduction of the density provided by the minimizing measure $ \musta |\log\eps| $. We proceed as follows: for any $ k \in \N $ larger than some given $ k_0 > 0 $, we set
	\beq
		\label{varrhok}
		\varrho_k : = \frac{k}{\sqrt{|\log\eps|}},
	\eeq
	and denote by $ \mathcal{C}_k $ and $ N_k $ the circle $ \partial \B(\varrho_k) $ and the number of equidistributed points we will put on $ \mathcal{C}_k $ respectively. Recalling \eqref{ms}, 
	we put uniformly distributed points on $ \mathcal{C}_k $ in such a way that the total number is (with $ \lfloor \: \cdot \: \rfloor $ standing for the integer part)
	\beq
		\label{nk}
		N_k = 
		\lf\lfloor 2 \pi \sqrt{|\log\eps|} \varrho_k \ms(\varrho_k) \ri\rfloor,
	\eeq
	provided $ \mathcal{C}_k \subset \supp(\musta) $.We denote by $ \kk \subset \N $ the set of integers $ k $ such that $ \mathcal{C}_k \subset \supp(\musta) $, so that one clearly has
	\beq
		\label{ub kk}	
		\sup_{k \in \kk} k \leq \rrs \sqrt{|\log\eps|},	\qquad	\sup_{k \in \kk} \lf| \mathbf{a}_{i_k} \ri| \leq \rrs,
	\eeq
	thanks to \eqref{rrs}. 
	The lower bound  $ k \geq k_0 $ is needed in order to ensure that on the smallest circle there are sufficiently many points (e.g., $ N_k \geq 4 $ for any $ k \geq k_0 $). 
	The point distribution naturally provides a decomposition of $ \mathcal{C}_k $ into sectors $ \vartheta_{i_k} \leq \vartheta \leq \vartheta_{i_k+1} $, $ i_k = 1, \ldots, N_k $, centered at the points $ \mathbf{a}_{i_k} $ and of equal arc length given by $ |\mathcal{C}_k|/N_k $, i.e., for $ k \in \kk $,
	\beq
		\label{Thetak}
		 \Theta_k : = \vartheta_{i_k+1} - \vartheta_{i_k} = 2\pi/N_k.
	\eeq
	The total number of points is of order $ \OO(|\log\eps|) $:
	\beq
		\label{total number}
		N_{\eps} = \sum_{k \in \kk} N_k = \OO(|\log\eps|).
	\eeq
	Moreover	
	$ \nutrial $ satisfies the required property, namely $ \nutrial \simeq |\log\eps| \musta $. More precisely 
	for any radial test function $ \Phi(r) \in C^1(\B(\rtf)) $,
	\bml{
 		\label{musta approx}
 		\int_{\B(\rtf)} \diff \rv \: \Phi(r) \nutrial = \lf(1 + \OO(\eps) \ri) \sum_{k \in \kk} \sum_{i_k = 1}^{N_k} \Phi\lf(a_{i_k}\ri)	 =  \lf(1 +\OO(\eps) \ri) \sum_{k \in \kk} N_k \Phi\lf(\varrho_k\ri) \\
		= 2 \pi \lf(1 +\OO(\eps) \ri) \sqrt{|\log\eps|} \sum_{k \in \kk}  \varrho_k \ms(\varrho_k) \Phi\lf(\varrho_k\ri) + \OO\lf(\sqrt{|\log\eps|}\ri)	\\
		=  |\log\eps| \int_{\supp(\musta)} \diff \rv \: \Phi(r)  \ms(r) + \OO(\sqrt{|\log\eps|})	=  |\log\eps| \int_{\R^2} \Phi \musta + \OO\lf(\sqrt{|\log\eps|}\ri),
	}
	where we have used the one-dimensional Riemann sum approximation, 
	\bdm
		\bigg| 2\pi |\log\eps|^{-1/2} \sum_{k \in \kk} \varrho_k f(\varrho_k) - \int_{\supp(\musta)} \diff \rv \: f(r) \bigg| \leq C \lf\|  f^{\prime}(r) \ri\|_{\infty} |\log\eps|^{-1/2},
	\edm
	 the fact that by hypothesis $ \lf\| \nabla \Phi \ri\|_{\infty} \leq C $ and the gradient estimate
	\beq
 		\lf\| \ms^{\prime}(r) \ri\|_{L^{\infty}(\supp(\musta))} \leq \half \sup_{r \leq \rrs} \lf| \partial^3_r \log \tfm \ri| = \half \sup_{r \leq \rrs} \tx \frac{\tfm^2 \tfm^{\prime\prime\prime} + 2 \lf( \tfm^{\prime} \ri)^3 - 3 \tfm \tfm^{\prime} \tfm^{\prime\prime}}{\tfm^3} \leq C,
	\eeq
	by \eqref{rrs}.

\subsection{Evaluation of the reduced energy $\E[\vtrial]$}

	Now we are going to estimate each term of the energy $ \E[\vtrial] $ separately. 

We start by considering the kinetic term: Noticing that one has
	\bdm
		\lf| \nabla \vtrial \ri|^2 = \ceps^2 \lf| \nabla \xi \ri|^2 + \ceps^2 \xi^2 \lf| \nabla \phitrial \ri|^2 \leq \lf(1 +  \OO(\eps^{2} |\log\eps|) \ri) \lf( \lf| \nabla \xi \ri|^2 +  \xi^2 \lf| \nabla \phitrial \ri|^2 \ri),
	\edm
	by \eqref{ceps est}, we obtain the bound
	\bml{
 		\label{kinetic ub 1}
 		\int_{\R^2} \diff \rv \: g^2 \lf|\nabla \vtrial \ri|^2 \leq \lf(1 + \OO(\eps^{2} |\log\eps|) \ri) \bigg[ \int_{\R^2} \diff \rv \: g^2 \lf| \nabla\xi \ri|^2 + \int_{\B(\rbl)\setminus\cup\B(\ai,\eps)} \diff \rv \: g^2 \lf| \nabla \phitrial \ri|^2 \bigg] 	\\
		\leq \lf(1 + \OO(\eps^{2} |\log\eps|)\ri) \int_{\B(\rsm)} \diff \rv \: g^{2} \tfm^{-2} \lf| \nabla h_{\nutrial} \ri|^2 + C \neps^2 \int_{\B(\rbl)\setminus\B(\rsm)} \diff \rv \: r^{-2} g^{2} +
 \OO(|\log\eps|)	\\
		 \leq \lf(1 + \OO(\eps^{2/3} |\log\eps|^{2/3}) \ri) \int_{\B(\rsm)} \diff \rv \: \tfm^{-1} \lf| \nabla h_{\nutrial} \ri|^2 + \OO(|\log\eps|) ,
	}
	where we have used the upper bounds \eqref{est xi gradients}, the estimate of the total number of points \eqref{total number} and \eqref{g point est} for the term located in $\B(\rsm)$. In the boundary layer $ \rsm \leq r \leq \rbl $ we have exploited the explicit expression for $ \nabla \phitrial $ \eqref{phi equation 3} and the pointwise estimate \eqref{g exp decay}. 
	We estimate now the interaction term, i.e. the last term in \eqref{reduced energy}:
	\bml{
 		\label{interaction ub 1}
 		\eps^{-2} \int_{\R^2} \diff \rv \: g^4 \lf( 1 - \lf| \vtrial \ri|^2 \ri)^2	\\
		\leq  \eps^{-2} \int_{\R^2 \setminus \B(\rbl - \eps^{2/3} |\log\eps|^{-1})} \diff \rv \: g^4 + \eps^{-2} \int_{\B(\rbl)} \diff \rv \: g^4 \lf( 1 - \lf| \vtrial \ri|^2 \ri)^2 	\\	
		\leq  C \eps^{-2} (\ceps^2 -1)^2 + \OO(\eps^{\infty}) \leq \OO(\eps^2|\log\eps|^2),
	}
	thanks to the exponential decay of $ g $ \eqref{g exp decay} and \eqref{ceps est}.

	The angular momentum term (second term in \eqref{reduced energy}) can be integrated by parts:
	\beq
 		\label{angular est 1}
 		- \Omega \int_{\R^2} \diff \rv \: g^2 \rv^{\perp} \cdot \lf(i \vtrial, \nabla \vtrial\ri) = |\log\eps| \int_{\B(\rbl)} \diff \rv \: F_{g}(r) \: \curl  \lf(i \vtrial, \nabla \vtrial\ri),
	\eeq
	where 
	\beq
		F_{g}(r) : = - \Omega_0 \int_{r}^{\rbl} \diff t \: t \: g^2(t).
	\eeq
	Note that in the integration by parts the boundary term vanishes thanks to the fact that $ F_g $ vanishes there. 
	 In fact in the boundary layer $ \B(\rbl) \setminus \B(\rsm) $ 
	 it is more convenient to integrate the r.h.s. of \eqref{angular est 1} by parts back: 
	\bml{
 		\label{angular est 2}
		|\log\eps| \int_{\B(\rbl)\setminus\B(\rsm)} \diff \rv \: F_g(r) \curl  \lf(i \vtrial, \nabla \vtrial\ri)	\\
		 = - \Omega_1 |\log\eps| \int_{\B(\rbl)\setminus\B(\rsm)} \diff \rv \: g^2 \rv^{\perp} \cdot \lf(i \vtrial, \nabla \vtrial\ri)	- |\log\eps| \int_{\partial \B(\rsm)} \diff \sigma \: F_g(\rsm) \lf(i \vtrial, \partial_\tau \vtrial\ri)	\\
		\leq  \eps^{2/3} \int_{\B(\rbl)\setminus\B(\rsm)} \diff \rv \: g^2 \lf| \nabla \phitrial \ri|^2 + C \eps^{-2/3} |\log\eps|^2 \int_{\B(\rbl)\setminus\B(\rsm)} \diff \rv \: r^2 g^2 +  \lf|F_g(\rsm)\ri| \OO(|\log\eps|)
		\\
		\leq \eps^{2/3} \int_{\B(\rbl)} \diff \rv \: g^2 \lf| \nabla \phitrial \ri|^2 + \OO(\eps^{2/3} |\log\eps|^{10/3})
	}
	by Cauchy-Schwarz inequality and the pointwise estimate \eqref{g exp decay}. 
	The same estimate also implies that $ |F_g(\rsm)| \leq C \eps^{4/3} |\log\eps|^{4/3} $, which has been used above as well as the fact that
	\beq
		\int_{\partial \B(\rsm)} \diff \sigma \: \lf(i \vtrial, \partial_\tau \vtrial\ri) = \ceps^2 \int_{\partial \B(\rsm)} \diff \sigma \: \partial_\tau \phitrial = 2\pi \ceps^2 \deg\lf\{\phitrial, \partial \B(\rsm) \ri\} = \OO(|\log\eps|),
	\eeq
	by \eqref{phase pro 2} and \eqref{total number}. The r.h.s. of \eqref{angular est 2} can then be safely incorporated in the remainder in \eqref{kinetic ub 1} and it is sufficient to consider the r.h.s. of \eqref{angular est 1} only inside $ \B(\rsm) $,
	where $ \xi_0 = 1 $ and 
	\beq
		\curl  \lf(i \vtrial, \nabla \vtrial\ri) = \ceps^2 \curl \lf( \nabla \phitrial \ri) = 0,		\quad	\mbox{in } \bigg( \bigcup_{i=1}^{N_{\eps}} \B(\ai,2\eps) \bigg)^{c} \bigcap \B(\rsm),
	\eeq
	so that we have 
	\bml{
 		\label{angular est 3}
		|\log\eps| \int_{\B(\rsm)} \diff \rv \: F_g(r) \curl  \lf(i \vtrial, \nabla \vtrial\ri)
		= |\log\eps| \int_{\cup \B(\ai,2\eps)} \diff \rv \: F_{g}(r) \: \curl  \lf(i \vtrial, \nabla \vtrial\ri) 	\\
		= \lf( 1 + \OO(\eps) \ri) |\log\eps| \sum_{i =1}^{N_{\eps}} F_g(a_i) \int_{\B(\ai,2\eps)} \diff \rv \: \curl \lf( \xi_v^2 \nabla \phitrial \ri)	\\
		= \lf( 1 + \OO(\eps) \ri) |\log\eps| \sum_{i =1}^{N_{\eps}} F_g(a_i) \int_{\partial \B(\ai,2\eps)} \diff \sigma \: \partial_{\tau} \phitrial  	
		\leq 2\pi |\log\eps| \sum_{i =1}^{N_{\eps}} F_g(a_i) + \OO(\eps |\log\eps|^2),
	}
	where we have used 
	the bound $ \| \nabla F_g \|_{\infty} \leq C $ to estimate 
	\beq
		\sup_{\rv \in \B(\ai,2\eps)} \lf| F_g(r) - F_g(a_i) \ri| \leq C \eps \lf| F_g(a_i) \ri|,
	\eeq
	since
	\beq
 		\label{lb Fg}
		\lf|F_g(a_i)\ri| \geq C \int_{a_i}^{\rsm} \diff t \: t \: g^2(t) \geq C \lf( 1 - \OO(\eps^{2/3} |\log\eps|^{2/3}) \ri) \int_{\rrs}^{\rsm} \diff t \: \tfm(t) \geq C > 0,
	\eeq
	thanks to \eqref{rrs}, \eqref{ub kk} and \eqref{g point est}. The replacement of $ F_g $ by $ \tfpot $ can be done in a similar way:
	\beq
		\lf| F_g(a_i) - \tfpot(a_i) \ri| \leq C \int_{r}^{\rbl} \diff t \: t \: \lf|g^2(t) - \tfm(t)\ri| \leq \OO(\eps^{2/3}|\log\eps|^{2/3}),
	\eeq
	by \eqref{ub kk} and \eqref{g point est}, so that using \eqref{lb Fg} one obtains 
	\bdm
		\lf| F_g(a_i) - \tfpot(a_i) \ri| \leq \OO(\eps^{2/3}|\log\eps|^{2/3}) \tfpot(a_i),	\quad	\mbox{for any } i = 1, \ldots, \neps.	
	\edm

	Combining \eqref{kinetic ub 1} and \eqref{interaction ub 1} with \eqref{angular est 1}, \eqref{angular est 2} and \eqref{angular est 3}, we obtain
	\beq
		\label{energy ub 1}
		\E\lf[\vtrial\ri] \leq  \lf(1 +\OO(\eps^{2/3} |\log\eps|^{2/3}) \ri) \int_{\B(\rsm)}  \frac{1}{2\tfm} \lf| \nabla h_{\nutrial} \ri|^2  + 2\pi |\log\eps| \sum_{i =1}^{N_{\eps}} F_g(a_i) + \OO(|\log\eps|),
	\eeq
	so that in order to complete the proof it suffices to estimate the first term on r.h.s. of the expression above. To this purpose we first note that the potential
	\beq
		\htilde : = h_{\nutrial} + \neps \tfm(\rsm) \log\lf(\rsm/\rbl\ri),
	\eeq
	solves 
	\beq
		\label{eq nutrial 2}
		\begin{cases}
			- \nabla \lf( \tfm^{-1} \nabla \htilde \ri) = \nutrial,		&	\mbox{in } \B(\rsm),	\\
			\htilde = 0,							&	\mbox{on } \partial \B(\rsm),
		\end{cases}
	\eeq
	since $ h_{\nutrial} $ is continuous in the whole of $ \B(\rbl) $ and for $ r \geq \rsm $ it has the explicit expression \eqref{hnutrial outside}. Obviously one also has
	\bdm
		\int_{\B(\rsm)} \diff \rv \: \frac{1}{\tfm} \lf| \nabla h_{\nutrial} \ri|^2 = \int_{\B(\rsm)} \diff \rv \: \frac{1}{\tfm} | \nabla \htilde|^2.
	\edm
	 It is now convenient to introduce the Green function
	\beq
		\G(\xv,\yv) : = A^{-1}(\xv,\yv),	\qquad	A : = - \nabla \lf( \tfm^{-1} \nabla \ri),
	\eeq
	i.e., the integral kernel of the inverse of the differential operator $ A $. More precisely $ \G $ solves the differential equation (we postpone to  Lemma \ref{G est: lem} below the proof of the required properties of $\G$)
	\beq
 		\label{Green definition}
		\begin{cases}
			- \nabla_{\xv} \lf[ \tfm^{-1}(x) \nabla_{\xv} \G(\xv,\yv) \ri] = \delta(\xv - \yv),	&	\mbox{for } \xv, \yv \in \B(\rsm),	\\
			\G(\xv,\yv) = 0,	&	\mbox{for } \xv \in \partial \B(\rsm).
		\end{cases}
	\eeq
	By means of $ \G $, we can rewrite the first term on the r.h.s. of \eqref{energy ub 1} as
	\bml{
 		\label{vortex interaction 1}
		\int_{\B(\rsm)}  \frac{1}{\tfm} \lf| \nabla h_{\nutrial} \ri|^2 = \int_{\B(\rsm)} \int_{\B(\rsm)} \: \G(\xv,\yv) \nutrial(\xv) \nutrial(\yv)	\\
		= \frac{4}{\eps^4} \sum_{i = 1}^{\neps} \int_{\B(\ai,\eps)} \diff \xv \int_{\B(\ai,\eps)} \diff \yv \: \G(\xv,\yv) +  \frac{4}{\eps^4} \sum_{\underset{i \neq j}{i,j= 1}}^{\neps} \int_{\B(\ai,\eps)} \diff \xv \int_{\B(\aj,\eps)} \diff \yv \: \G(\xv,\yv).
	}

	For the estimate of the diagonal term we can use the inequality \eqref{G est} proven in Lemma \ref{G est: lem} below, which yields
	\bml{
		\label{vortex interaction 2}
		\frac{4}{\eps^4} \sum_{i = 1}^{\neps} \int_{\B(\ai,\eps)} \diff \xv \int_{\B(\ai,\eps)} \diff \yv \: \G(\xv,\yv) \leq - \frac{2}{\pi \eps^4} \sum_{i = 1}^{\neps} \int_{\B(\ai,\eps)} \diff \xv \int_{\B(\ai,\eps)} \diff \yv \:\tfm(y) \log|\xv - \yv| + \OO(|\log\eps|)	\\
		\leq - \lf(1 + \OO(\eps) \ri) \frac{2}{\pi \eps^4} \sum_{i = 1}^{\neps} \tfm(a_i) \int_{\B(\ai,\eps)} \diff \xv \int_{\B(\ai,\eps)} \diff \yv \:  \log|\xv - \yv| + \OO(|\log\eps|)	\\
		\leq  2 \pi  \sum_{i = 1}^{\neps} \tfm(a_i) |\log\eps| + \OO(|\log\eps|),
	}
	where we have used the properties of harmonic functions and in particular the mean value theorem to compute
	\bml{
		- \int_{\B(\ai,\eps)} \diff \xv \int_{\B(\ai,\eps)} \diff \yv \:  \log|\xv - \yv| 	\\
		= - \half \int_0^{2\pi} \diff \vartheta \int_0^{2\pi} \diff \vartheta^{\prime} \int_0^{\eps} \diff \varrho \varrho\int_0^{\eps} \diff \varrho^{\prime} \varrho^{\prime} \: \log \lf[ \varrho^2 + {\varrho^{\prime}}^2 - 2 \varrho \varrho^{\prime} \cos\lf( \vartheta - \vartheta^{\prime} \ri) \ri] 	\\
		= - 4 \pi^2 \int_0^{\eps} \diff \varrho \varrho \lf[ \int_0^{\varrho} \diff \varrho^{\prime} \varrho^{\prime} \: \log \varrho + \int_{\varrho}^{\eps} \diff \varrho^{\prime} \varrho^{\prime} \: \log \varrho^{\prime} \ri] = \pi^2 \eps^4 |\log\eps| + \half \pi^2 \eps^4. 
	}
	To estimate of the difference $ \tfm(r) - \tfm(a_i) $, we have used the bounds $ \| \nabla \tfm \|_{\infty} \leq C $ and 
	\beq
		\tfm(a_i) \geq C > 0,	\quad	\mbox{for any } i = 1, \ldots, \neps,
	\eeq
	which is a simple consequence of the explicit expression of $ \tfm $, \eqref{rrs} and \eqref{ub kk}.

	The off-diagonal term on the other hand can be dealt with by using the inequality \eqref{G grad est} proven in Lemma \ref{G est: lem} and a Riemann sum approximation. The distribution of points $ \{ \ai \} $ identifies a cell decomposition of $ \B(\rrs) $: one can associate with each sector a two-dimensional cell $ \cell_{i_k} $ as
	\beq
		\label{cell ik}
		\cell_{i_k} : = \lf\{ \rv = (r,\vartheta) : \:  \vartheta \in \lf[ \vartheta_{i_k} - \half\Theta_k , \vartheta_{i_k} + \half \Theta_k \ri], r \in \lf[ (k - \half) |\log\eps|^{-1/2}, (k + \half) |\log\eps|^{-1/2} \ri] \ri\},
	\eeq
	with volume
	\beq
		\lf| \cell_{i_k} \ri| = \frac{\varrho_k \Theta_k}{\sqrt{|\log\eps|}}  = \frac{2 \pi \varrho_k}{ N_k \sqrt{|\log\eps|}} \geq  \frac{1}{\ms(\varrho_k)|\log\eps|},
	\eeq
	by \eqref{varrhok}, \eqref{nk} and \eqref{Thetak}. Given a $i_k$ it will be convenient to denote $n(i_k)$ the set of indices $j_h$ such that $\cell_{j_h}$ is a nearest neighbor of $\cell_{i_k}$ in the cell decomposition. Indeed, $\G(\xv,\yv)$ is singular for small $|\xv-\yv|$ and it is therefore useful to distinguish those pairs of cells $\cell_{i_k},\cell_{j_h}$ for which the distance $|\xv-\yv|$ can be arbitrarily small for $\xv\in \cell{i_k}, \yv \in \cell_{j_h}$.
	
	We first write, using \eqref{G grad est} and the fact that the distance between $\avi$ and $\avj$ is at least of order $|\log \ep| ^{-1/2}$ for $i\neq j$, 
	\beq
		\label{vortex interaction 3}
 		\frac{4}{\eps^4} \sum_{\underset{i \neq j}{i,j= 1}}^{\neps} \int_{\B(\ai,\eps)} \diff \xv \int_{\B(\aj,\eps)} \diff \yv \: \G(\xv,\yv) = (2\pi)^2  \sum_{\underset{i \neq j}{i,j= 1}}^{\neps} \G(\ai,\aj) + \OO(\eps |\log\eps|^2 \log|\log\eps|).
	\eeq
        Then we neglect the neighboring points in the sum, using \eqref{G est} and the fact that $|\avik-\avjh|$ is of order $|\log \ep| ^{-1/2}$ for $j_h \in n(i_k)$:
        \[
         \sum_{\underset{i \neq j}{i,j= 1}}^{\neps} \G(\ai,\aj) = \sum_{k,h \in \mathcal K}  \sum_{i_k= 1}^{N_k}  \sum_{\underset{ j_h\notin n(i_k)}{j_h= 1}}^{N_h} \G(\avik,\avjh) + \OO(|\log \ep|\log |\log \ep|).
        \]
	Next we use \eqref{G grad est} to write (here its is crucial to have excluded from the sum the nearest neighbors of each point to have lower bounds on $|\xv-\yv|$)
	\begin{multline}
		 (2\pi)^2  \sum_{k,h \in \mathcal K}  \sum_{i_k= 1}^{N_k}  \sum_{\underset{ j_h\notin n(i_k)}{j_h= 1}}^{N_h} \G(\avik,\avjh) =  \sum_{k,h \in \mathcal K}  \frac{(2\pi)^2}{\lf| \cell_{i_k} \ri| \lf| \cell_{j_h} \ri|} \sum_{i_k= 1}^{N_k}  \sum_{\underset{ j_h\notin n(i_k)}{j_h= 1}}^{N_h} \int_{\cell_{i_k}} \diff \xv \int_{\cell_{j_h}} \diff \yv \: \G(\xv,\yv) \\+ \OO(|\log\eps|^{3/2} \log|\log\eps|).
	\end{multline}
	Now, using \eqref{G est} again, one can see that neighboring cells do not weigh too much in the integrals:
	\[
	 \sum_{k,h \in \mathcal K}  \frac{(2\pi)^2}{\lf| \cell_{i_k} \ri| \lf| \cell_{j_h} \ri|} \sum_{i_k= 1}^{N_k}  \sum_{\underset{ j_h\in n(i_k)}{j_h= 1}}^{N_h} \int_{\cell_{i_k}} \diff \xv \int_{\cell_{j_h}} \diff \yv \: \G(\xv,\yv) = \OO(|\log \ep|\log |\log \ep|),
	\]
	which allows to conclude, gathering the above equations and using a Riemann sum approximation again, 
	\begin{multline}
	 \frac{4}{\eps^4} \sum_{\underset{i \neq j}{i,j= 1}}^{\neps} \int_{\B(\ai,\eps)} \diff \xv \int_{\B(\aj,\eps)} \diff \yv \: \G(\xv,\yv)= |\log\eps|^2  \int_{\B(\rrs)} \int_{\B(\rrs)} \: \G(\xv,\yv) \musta(x)\musta(y) \\+\OO(|\log \ep| ^{3/2} \log |\log \ep|).
	\end{multline}

	We can now go back to \eqref{energy ub 1} and replace \eqref{vortex interaction 2} and \eqref{vortex interaction 3}, finally obtaining
	\bml{ 
 		\label{energy ub 2}
		\E\lf[\vtrial\ri] \leq 
		2 \pi |\log\eps|  \sum_{i = 1}^{\neps} \lf[ \half \tfm(a_i) + \tfpot(a_i) \ri] +  (2\pi)^2  \sum_{\underset{i \neq j}{i,j= 1}}^{\neps} \G(\ai,\aj)	 + \OO(|\log\eps|)	\\
		\leq   2 \pi |\log\eps|  \sum_{i = 1}^{\neps} \tfH(a_i) + |\log\eps|^2  \int_{\B(\rrs)} \int_{\B(\rrs)} \: \G(\xv,\yv) \musta(x)\musta(y) + \OO(|\log\eps|^{3/2} \log|\log\eps|) \\
		\leq  |\log\eps|^2 \int_{\B(\tfr)} \diff \rv \: \lf\{ \frac{1}{\tfm} \lf| \nabla h_{\musta} \ri|^2 + \tfH(r) \ms(r) \ri\} + \OO(|\log \ep| ^{3/2} \log |\log \ep|)\\
		 = \tfI |\log\eps|^2 + \OO(|\log\eps|^{3/2} \log|\log\eps|),
	}
	where we have used  \eqref{musta approx} for the term involving $\tfH$. Recalling \eqref{reduced energy}, the proof of Proposition \ref{upper bound: pro} is now complete.
	
	\medskip

We conclude this section by proving some technical results about the Green function $ \G(\xv,\yv) $ used in the proof above. We recall that $\G$ is defined as the solution of (see \eqref{Green definition})
\bdm
	\begin{cases}
		- \nabla_{\xv} \lf[ \tfm^{-1}(x) \nabla_{\xv} \G(\xv,\yv) \ri] = \delta(\xv - \yv),	&	\mbox{for } \xv, \yv \in \B(\rsm),	\\
		\G(\xv,\yv) = 0,	&	\mbox{for } \xv \in \partial \B(\rsm).
	\end{cases}
\edm

\begin{lem}[\textbf{Properties of the Green function $ \G(\xv,\yv) $}]
	\label{G est: lem}
	\mbox{}	\\	
	There exists a unique solution $ \G(\xv,\yv) $ to \eqref{Green definition} in $ W^{1,p}(\B(\rsm) \otimes \B(\rsm)) $ for any $ 1 \leq p < 2 $. It is positive and symmetric and 
	\beq
		\label{G est}
		\sup_{\xv,\yv \in \B(\rsm)} \lf| \G(\xv,\yv) + \frac{1}{2\pi} \tfm(y) \log|\xv - \yv| \ri| \leq C
	\eeq
	for some fixed constant $C>0$.
	Moreover, for any $i,j\in \N$ and $\yv \in \cell_j$, 
	\beq
		\label{G grad est}
		\sup_{\xv \in \cell_{i}} \lf| \nabla_x \G(\xv,\yv) \ri| \leq C \left( \sup_{\xv \in \cell_{i}} \lf| \log \lf| \xv - \yv \ri| \ri| +1 \right),
	\eeq
	where the cell $ \cell_{i} $ is defined in \eqref{cell ik}.
\end{lem}

\begin{proof}
	Most of the properties of $ \G $ follow from standard arguments applied to the analysis of the elliptic operator $ A $ (see, e.g., \cite{St}). Moreover analogous results including \eqref{G est} have already been proven in \cite[Lemma 2.4]{R} and \cite[Lemma 4.5]{AAB} (see also \cite{ASS}) and we skip here the details for the sake of brevity, only stressing that the crucial ingredient for \eqref{G est} is the fact that $ \| \nabla \tfm \|_{\infty} \leq C $.

	The second inequality is a consequence of the equation for $ \G $, which for $ \xv \in \cell_i $ and $ \yv \in \cell_j $, $i\neq j$, can be rewritten
	\bdm		
		\nabla_{\xv} \tfm(x) \cdot \nabla_{\xv} \G(\xv,\yv) - \tfm(x)  \Delta_{\xv} \G(\xv,\yv) = 0,			\edm
	together with the Gagliardo-Nirenberg inequality (see, e.g., \cite[Theorem 1]{N})
	\bdm
		\lf\| \nabla_{\xv}\G \ri\|_{\infty} \leq C \lf( \lf\| \Delta_{\xv}\G \ri\|^{1/2}_{\infty}  \lf\| \G \ri\|^{1/2}_{\infty} + \lf\|  \G \ri\|_{\infty} \ri),
	\edm
	where the suprema are taken with respect to $ \xv \in \celli $ and for almost every $ \yv \in \cell_j $. Combining this with the inequality above we get $  \lf\| \nabla_{\xv}\G \ri\|_{\infty} \leq C  \lf\|  \G \ri\|_{\infty} $ since $ \tfm $ is uniformly bounded from below inside the support of $ \musta $. Using \eqref{G est} one then obtains \eqref{G grad est}.	
\end{proof}

\section{Energy Lower Bound and Convergence of the Vorticity Measure}\label{sec:lower bound}

In this section we complete the proof of Theorem \ref{teo:gse asympt} by proving the following:

\begin{pro}[\textbf{GP energy lower bound}]
	\label{pro:lower bound}
	\mbox{}	\\
	If $ \Omega = \Omega_0 |\log\eps| $, with $ \Omega_0 > \Omega_1 $ as $ \eps \to 0 $, then
	\beq
		\gpe \geq \hgpe + \tfI |\log\eps|^2 \left( 1 + \OO\left( \frac{\log |\log \ep|}{|\log \ep|} \right) \right).
	\eeq
\end{pro}

In the course of the proof we will gather the estimates leading to the conclusion of Theorem \ref{teo:vorticity}, in particular the rigorous version of \eqref{eq:sk final}. The strategy we employ makes use of several classical techniques which have already been used in similar contexts, see, e.g., \cite{AAB,ABM,CRY,IM1,R}. It originates from the GL theory of type-II superconductors (see \cite{SS2} and references therein).

\subsection{Preliminary Steps}

We start as usual with the energy decoupling 
\begin{equation}\label{eq:low decouple}
\gpe = \hgpe + \int_{\R ^2} \diff \rv \:  \eg(u),
\end{equation}
where we have denoted 
\begin{equation}\label{eq:low u}
u := \frac{\gpm}{g}
\end{equation}
and the reduced energy density is
\begin{equation}\label{eq:low ener density}
\eg (u) : = \half g^2 |\nabla u | ^2 - \Omega g ^2 \rv^{\perp} \cdot \left(iu, \nabla u \right)+ \frac{g^4}{\ep ^2} \left( 1 -|u| ^2\right)^2,
\end{equation}
with, as usual,
\[
 \left( iu,\nabla u \right) = \frac{i}{2}\left( u\nabla u ^* - u ^* \nabla u\right).
\]
For a lower bound it is possible to neglect the contribution to the reduced energy in the region $\R ^2\setminus \B (R_+)$ with 
\begin{equation}\label{eq:rbl}
 R_+ : = \tfr + \ep ^{1/4}.
\end{equation}
Indeed, the only possibly negative term can be bounded below by a simple trick 
\begin{equation}\label{eq:low trick}
  - \Omega\int_{\R^2  \setminus \B (R_+)}  \diff \rv \: g ^2 \cdot \left(iu, \nabla u \right) \geq  
- \frac{1}{4} \int_{\R^2  \setminus \B (R_+)} \diff \rv \:  g ^2 \left| \nabla u \right| ^2 - \Om ^2 \int_{\R^2  \setminus \B (R_+)} \diff \rv \: g ^2 |u| ^2
\end{equation}
and in $\R^2  \setminus \B (R_+)$ we can use the exponential smallness (see Proposition \ref{pro:gpm pointwise}) of $\gpm = gu$ to bound the second term and obtain 
\begin{equation}\label{eq:low cut domain 1}
\int_{\R ^2} \diff \rv \: \eg(u) \geq \int_{\B (R_+)} \diff \rv \: \eg(u) - \OO (\ep ^{\infty})
\end{equation}
where positive but useless contributions have been dropped.

Note for later use that the trick \eqref{eq:low trick} used on the whole of $\R^2$ together with the trivial bound
\[
\E[u]\leq 0, 
\]
which follows by taking $g$ as a trial state for $\gpf$, imply
\begin{equation}\label{eq:low bound F}
\int_{\R ^2} \diff \rv \: \lf\{ \frac{1}{4} g ^2 |\nabla u| ^2 + \frac{g ^4}{\ep ^2} \left(1-|u| ^2 \right)^2 \ri\} \leq C |\log \ep| ^2,
\end{equation}
because of the normalization of $g^2 |u| ^2 = |\gpm| ^2$.

We now define
\begin{equation}\label{eq:low pot}
F(r) := -\Omega_0 \int_{r} ^{R_+} \diff t \: t \: g ^2 (t), 
\end{equation}
which satisfies $|\log \ep |\nablap F = g ^2 \Omega \rvp $ in $\B(R_+)$ and $F(R_+) = 0$. We can then integrate by parts the momentum term in \eqref{eq:low cut domain 1} to obtain 
\begin{equation}\label{eq:low ipp}
\int_{\B (R_+)} \diff \rv \: \eg(u) = \int_{\B(R_+)}\diff \rv  \lf\{ \frac{1}{2} g^2 |\nabla u| ^2 + |\log \ep|  F \curl (iu,\nabla u) + \frac{g ^4}{\ep ^2 }\left(1- |u | ^2 \right) ^2\ri\}. 
\end{equation}
We will soon construct vortex balls to bound this energy from below by the energy of vortices of $u$. There is however one more reduction we need to perform before: as it is well-known, the vortex balls construction is feasible only in the region where the matter density is large enough. We thus split the domain again, introducing for some constant $\cbulk > 0 $ to be chosen later on ($\rbulk$ is the radius appearing in the statement of Theorem \ref{teo:gse asympt})
\begin{equation}\label{eq:low radii}
\rbulk : = R_+ - \cbulk \Omega ^{-1},\qquad \Rc := R_+ - \half \cbulk \Om ^{-1}. 
\end{equation}
Note that we have now four radii in our construction:
\[
\rbulk < \Rc < \tfr < R_+
\]
and in particular 
\begin{equation}\label{eq:Rcut}
\Rc \leq \tfr - C \Om ^{-1}, 
\end{equation}
which is a consequence of \eqref{eq:rbl} and \eqref{eq:low radii}. Using \eqref{g point est} and the explicit form of $\tfm$ (in particular the fact that it vanishes linearly on $\dd \B(\rtf)$), this implies
\begin{equation}\label{eq:dens BRc}
g^2 (r) \geq C |\log \ep| ^{-1},	\quad \mbox{for any } \rv \in \B(\Rc),
\end{equation}
and for some $ C > 0 $.

We now perform a smooth cut-off at $\Rc$ by introducing regular radial  functions $\chiin$ and $\chiout$ satisfying
\begin{equation}\label{eq:low chiinout}
\chiin + \chiout = 1 \mbox{ in } \R ^2, \quad \chiin (r) = 1 \mbox{ for } r \leq \rbulk, \quad \chiin (r ) = 0 \mbox{ for } r \geq \Rc.   
\end{equation}
Using \eqref{eq:low radii} and recalling that we consider the regime $\Om = \OO(|\log\ep|)$ we can impose
\begin{equation}\label{eq:grad chi}
\left| \nabla \chiin \right|\leq C |\log \ep| ^{-1},\qquad \left| \nabla \chiout \right|\leq C |\log \ep| ^{-1}  .
\end{equation}
The point of performing this cut-off is that we choose $\cbulk$ small enough in \eqref{eq:low radii} in such a way that
\beq
	\label{bound F outside}
|\log \ep| |F| \leq \frac{1}{4} g^2,	\quad		 \mbox{in } \B(R_+) \setminus \B(\rbulk).
\eeq
This is an easy consequence of the definition of $F$ \eqref{eq:low pot}: for any $ \rv \in  \B(R_+) \setminus \B(\rbulk) $ one has (recall that $ g $ is decreasing as discussed in the Appendix)
\beq
	\label{est F outside}
 	\lf| F(r) \ri| = \Omega_0 \bigg| \int_{r} ^{R_+} \diff t \: t \: g ^2 (t) \bigg| \leq \half \Omega_0 g^2(r) \lf(R_+^2 - r^2 \ri) \leq \half \cbulk R_+ |\log\eps|^{-1} g^2(r), 
\eeq
so that the above inequality holds true if one picks, e.g., $ \cbulk \leq \half R_+^{-1} $.
Using \eqref{bound F outside} and the trivial bound (the two line computation is in \cite[Lemma 3.4]{CPRY3})
\[
 |\curl (iu,\nabla u)| \leq |\nabla u| ^2,
\]
we see that 
$$\chiout\left( \frac{1}{4}g ^2 |\nabla u| ^2 + |\log \ep| F \curl(iu,\nabla u) \right) \geq 0.$$
This leads to 
\begin{equation}\label{eq:low cut domain 2}
\int_{\B(R_+)} \diff \rv \: \eg (u) \geq \int_{\B(R_+)} \diff \rv \left\{\chiin \left[ \half g ^2 |\nabla u| ^2 + F \curl(iu,\nabla u) \right] + \frac{1}{4}\chiout  g ^2 |\nabla u| ^2 + \displaystyle\frac{g ^4}{\ep ^2} (1-|u| ^2) ^2\right\}.
\end{equation}

\subsection{Evaluating the Individual Energy of Vortices}

We now state two classical results allowing to spot the vortices of the condensate and evaluate their self-energy.
The methods (vortex balls construction and Jacobian estimates) are classical tools by now and we thus skip the proofs. The reader can refer to \cite[Chapter 4 and 6]{SS2} for a self-contained presentation of the techniques and to \cite[Proposition 4.1]{IM1}, \cite[Proposition 4.1]{AAB}, \cite[Proposition 4.2 and 4.3]{CRY} for applications in contexts similar to the present one. The original references are \cite{J,JS2,Sa}.

The first result isolates the possible zeros and phase singularities, i.e., vortices, of $u$ in small balls. It is a consequence of the bound \eqref{eq:low bound F} and uses \eqref{eq:dens BRc}, which makes the reduction to $\B(\Rc)$ mandatory.  

\begin{pro}[\textbf{Vortex ball construction}]\label{pro:vortexballs}
 	\mbox{}\\
	There is a certain $\ep_0>0$ such that, for $\ep \leq \ep_0$ there exists a finite collection $ \{ \B_j \}_{j \in J} := \left\lbrace \B (\avj, \varrho_j)\right\rbrace_{j\in J}$ of disjoint balls with centers $ \avj $ and radii $ \varrho_j $ such that
	\begin{enumerate}
		\item $\left\lbrace \rv \in \B (\Rc): \: \left| |u(\rv)| - 1  \right| > |\log \ep| ^{-5}  \right\rbrace \subset \bigcup_{j\in J}  \B_j$,
		\item  $\displaystyle\sum_{j\in J} \varrho_j = |\log \ep |^{-10} $.		   
	\end{enumerate}
	Setting $d_j:= \dg \{ u, \partial \B_j \} $, if $ \B_j \subset \B (\Rc) $, and $d_j=0$ otherwise, we have the lower bounds
		\begin{equation}\label{eq:lowboundballs}
		 	\int_{\B_j}  \diff \rv \: g ^2 \left|\nabla u\right|^2 \geq \pi |d_j|   g^2 (a_j) \left| \log \ep \right| \left(1-C \frac{\log \left| \log \ep \right|}{\left|\log \ep\right|}\right)
		\end{equation}
	for some fixed constant $C>0$.
\end{pro}

Note that the error term $\OO (\log |\log \ep|)$ in \eqref{eq:lowboundballs} has two origins: one is the use of the growth and merging method \cite{Sa,J}. It could probably be improved by using more refined versions of the technique, as presented, e.g., in \cite{SS3}. The other is the approximation of $g^2$ in each vortex ball by its value at the center using \eqref{g point est} and the regularity of $\tfm$. The second error is much smaller than the first one and is thus absorbed in the $\OO (\log |\log \ep|)$ term. It is also much more intrinsic to our setting since it originates from the inhomogeneity of the matter density. 

Remark that there is some freedom in the definition of the vortex balls: we could take any negative power of $|\log \ep|$ at point 1 of the statement above, obtaining ball families such that the sum of the radii would be bounded by an arbitrary negative power of $|\log \ep|$. This would not affect the order of magnitude of the remainder in \eqref{eq:lowboundballs} and would improve some other remainders that will appear in the sequel. Since the error we make in \eqref{eq:lowboundballs} will necessarily be present in the final result, we see no point in over-optimizing other estimates and thus stick to the concrete choice made above (which by the way is the same as in \cite[Proposition 4.1]{IM1}).

We denote 
\begin{equation}\label{eq:low F(u)}
\F [u] : = \int_{\B(R_+)} \diff \rv \lf\{ \frac{1}{2} g^2 |\nabla u| ^2 + \frac{g ^4}{\ep ^2} \left(1-|u|^2 \right) ^2 \ri\}. 
\end{equation}
It is well-known \cite{JS2} that this quantity can be used to control the vorticity measure (recall our normalization choice)
\begin{equation}\label{eq:low vorticity}
\mu = |\log \ep| ^{-1} \curl \left( iu,\nabla u\right), 
\end{equation}
which is the content of the following result.

\begin{pro}[\textbf{Jacobian estimate}]\label{pro:jacest}
 	\mbox{}\\
	Let $\phi$ be any piecewise-$C^1$ test function with compact support 
	\[
	 {\rm supp}(\phi) \subset \B(\rbulk).  
	\]
	Let $\left\lbrace \B_j \right\rbrace_{j\in J} : = \lf\{ \B(\avj,r_j) \ri\}_{j \in J} $ be a collection of disjoint balls as in Proposition \ref{pro:vortexballs}. Setting $d_j:= \dg \{ u, \partial \B_j \} $, if $ \B_j \subset \B(\Rc) $, and $d_j=0$ otherwise, one has
	\begin{equation}\label{eq:JE}
		\bigg|\sum_{j\in J}  2 \pi d_j \phi (\avj)- |\log \ep|\int_{\B(\rbulk)} \phi \:  \mu \bigg| \leq  C |\log \ep| ^{-9} \left\Vert \nabla \phi \right\Vert_{L^{\infty}(\B(\rbulk))}  \F [u]
	\end{equation}	
	where the constant $C$ is independent of $\phi$.
\end{pro}

Taken together, these results imply
\begin{multline}\label{eq:low use balls}
\int_{\B(R_+)} \diff \rv \left\{\chiin \left( \frac{1}{2}g ^2 |\nabla u| ^2 + |\log \ep| F \curl \left( iu,\nabla u\right) \right) \right\}	\geq  \frac{1}{2} \int_{\B(R_+) \setminus \cup_{i\in I} \B_i}\diff \rv \: \chiin g ^2 |\nabla u| ^2	\\ + \sum_{j\in J} 2\pi |\log \ep|\chiin(a_j)\left( \half |d_i| g^2 (a_j) + d_i F(a_j)\right) \left(1 -C \frac{\log |\log \ep|}{|\log \ep|} \right) - \OO(|\log\eps| \log|\log\eps|),
\end{multline}
where the bound \eqref{eq:grad chi} has been used. The last remainder terms contains two contributions: the smaller one is due to \eqref{eq:JE} can be easily seen to be of order $ \OO(|\log\eps|^{-7}) $ by using the a priori bound \eqref{eq:low bound F} on $ \F[u] $. The second one is due to the factor $ - C \log|\log\eps| |\log\eps|^{-1} $ multiplying the (negative) vortex energy gain $  \sum 2\pi |\log \ep| d_i F(a_j) $, i.e.,
\bdm	 
	\log|\log \ep| \int_{\B(R_+)} \diff \rv \: F \curl \left( iu,\nabla u\right) = - \Omega_0 \log|\log\eps| \int_{\B(R_+)} \diff \rv \: g^2 \left( iu,\nabla u\right)
\edm
and it can be estimated precisely as in \eqref{eq:low trick}, again with the help of \eqref{eq:low bound F}, yielding an error term of order $ |\log\eps| \log|\log\eps| $. 

The purpose of the next section is to prove a lower bound to the first term of the right-hand side of this inequality. As we will see, it is this term that contains the energy due to the interaction between vortices.

\subsection{Evaluating the Interaction Energy of Vortices}

We now estimate the first term of the right-hand side of \eqref{eq:low use balls}. Following \cite{ABM,R} we first regularize the vorticity measure $\mu$, by introducing the modified current
\begin{equation}\label{eq:low jtilde}
\jt = \begin{cases}
(iu,\nabla u),	& \mbox{in } \B(\Rc) \setminus \bigcup_{j \in J} \B_j \\
0,			& \mbox{otherwise}. 
\end{cases}
\end{equation} 
The regularized vorticity measure $\mut$ is the one naturally associated with $\jt$ (renormalized as in \eqref{eq:intro mu}):
\begin{equation}\label{eq:low mutilde}
\mut := |\log \ep| ^{-1}\curl   \jt.
\end{equation}
It is more regular than the original $\mu$ in the sense that we have removed the singular part of the phase generated by vortices by neglecting in \eqref{eq:low jtilde} the part of the current inside the vortex balls. There is still of course a line singularity along the boundary of the balls because the current $\jt$ goes brutally to $0$ there, but this is much less problematic that the vortices lying inside the balls, which are point singularities.

The following lemma shows that the regularized vorticity measure is close to the original one in the $(C^1_c) ^*$ norm, allowing to control the error term due to the replacement of $ \mu $ by its regularization $ \mut $.

\begin{lem}[\textbf{Approximation by a regularized vorticity measure}]\label{lem:regul vortic}\mbox{}\\
There exists a constant $C>0$ such that, for any piecewise-$C^1$ test function $\phi$ with compact support 
	\[
	 {\rm supp}(\phi) \subset \B(\Rc),
	\] 
we have
\begin{equation}\label{eq:prop mutilde}
\bigg| \int_{\B(\Rc)} \left( \mu - \mut \right)\phi \bigg| \leq C |\log \ep| ^{-8} \left\Vert \nabla \phi\right\Vert_{L ^{\infty} (\B(\Rc))}.
\end{equation}
\end{lem}

\begin{proof}
By the definitions \eqref{eq:low vorticity} and \eqref{eq:low mutilde} 
\[
\bigg| \int_{\B(\Rc)} \left( \mu - \mut \right)\phi \bigg| = |\log \ep| ^{-1} \bigg| \int_{\B(\Rc)} \diff \rv \: \nablap \phi \cdot \left(\jv -\jt \right)\bigg| = |\log \ep| ^{-1}  \bigg| \int_{\cup_{j\in J} \B_j} \diff \rv \: \nablap \phi \cdot \jv \bigg|.
\]
Then, using \eqref{eq:dens BRc},
\begin{eqnarray*}
\bigg| \int_{\cup_{j\in J} \B_j}  \diff \rv \: \nablap \phi \cdot \jv \bigg| &\leq&  \left\Vert \nabla \phi \right\Vert_{L ^{\infty} (\B(\Rc))} \int_{\cup_{j\in J} \B_j} \diff \rv \: |u||\nabla u| \\
&\leq& C|\log \ep|  \left\Vert \nabla \phi \right\Vert_{L ^{\infty} (\B(\Rc))} \int_{\cup_{j\in J} \B_j}  \diff \rv \: g ^2 |u||\nabla u|\\
&\leq& C|\log \ep| \left\Vert \nabla \phi \right\Vert_{L ^{\infty} (\B(\Rc))} \bigg( \delta \int_{\cup_{j\in J} \B_j}  \diff \rv \: g^2 |u| ^2 +\delta^{-1} \int_{\R ^2}  \diff \rv \:  g^2 |\nabla u| ^2 \bigg)
\end{eqnarray*}
for some $\delta >0$. We now use the upper bound $g^2|u|^2 = |\gpm| ^2 \leq C$ (see Proposition \ref{pro:gpm pointwise}) and the upper bound on $\sum_j \varrho_j$ of Proposition \ref{pro:vortexballs}, 
to obtain
\[
\int_{\cup_{j\in J} \B_j}  \diff \rv \: g^2 |u| ^2 \leq C |\log \ep| ^{-20}. 
\]
Recalling the bound \eqref{eq:low bound F}, we also have
\[
\int_{\R ^2}  \diff \rv \: g ^2 |\nabla u| ^2 \leq C |\log \ep| ^2 
\]
and, choosing $\delta = |\log \ep| ^{11}$, we conclude the proof.
\end{proof}

We now define $h_{\mut}$ as the unique solution to 
\begin{equation}\label{eq:defi hmut}
\begin{cases}
		-\nabla \left( \frac{1}{g^2} \nabla  h_{\mut}\right) = \mut,	& 	\mbox{in } \B(\rbulk),	\\
		h_{\mut} = 0,								&	\mbox{on } \partial \B(\rbulk),
	\end{cases}  
\end{equation}
and claim that the following holds:

\begin{lem}[\textbf{Vortex interaction energy}]\label{lem:low int estimate}\mbox{}\\
Let $\left\lbrace \B_j \right\rbrace_{j\in J} : = \lf\{ \B(\avj,r_j) \ri\}_{j \in J} $ be a collection of disjoint balls as in Proposition \ref{pro:vortexballs}. We have
\begin{equation}\label{eq:low int energy}
\int_{\B(R_+) \setminus \cup_{j\in J} \B_j}  \diff \rv \: \chiin g ^2 |\nabla u| ^2 \geq \left(1-C|\log \ep|^{-5}\right)|\log \ep| ^2\int_{\B(\rbulk)}  \diff \rv \:  \frac{1}{g ^2} |\nabla h_{\mut}| ^2
\end{equation}
for some given constant $C>0$.
\end{lem}

\begin{proof}
The proof is the same as in \cite[Lemma 3.3]{R}, only a bit simpler because we work in a simply connected geometry. We recall that $\chiin = 1$ in $\B(\rbulk) \setminus \cup_{j\in J} \B_j$ and $|u|$ is close to $1$ there, according to point 1 in Proposition \ref{pro:vortexballs}. Thus  
\begin{eqnarray*}
\int _{\B(R_+) \setminus \cup_{j\in J} \B_j}  \diff \rv \: \chiin g^2 |\nabla u| ^2 &\geq& \int _{\B(\rbulk)  \setminus \cup_{j\in J} \B_j}  \diff \rv \: g^2 |\nabla u| ^2 \nonumber \\
&\geq& \left( 1 - C |\log \ep| ^{-5}\right) \int _{\B(\rbulk)  \setminus \cup_{j\in J} \B_j}  \diff \rv \: g^2 |u|^2 |\nabla u | ^2 \nonumber \\
&\geq & \left( 1 - C |\log \ep| ^{-5}\right)\int _{\B(\rbulk)  \setminus \cup_{j\in J} \B_j}  \diff \rv \: g^2 \left| (iu,\nabla u) \right| ^2 \nonumber \\
&=& \left( 1 - C |\log \ep| ^{-5}\right)\int _{\B(\rbulk)}  \diff \rv \: g^2 |  \jt | ^2.  
\end{eqnarray*}
By definition
\[
\curl \left(|\log \ep| ^{-1}\jt + \frac{1}{g^2} \nablap h _{\mut}\right) = 0,	\quad \mbox{in } H^{-1}\left(\B(\rbulk)\right),
\]
so that there exists $f\in H^{1} (\B(\rbulk))$ satisfying
\[
\jt = - |\log \ep| \frac{1}{g ^2} \nablap h_{\mut} + \nabla f.
\]
Using the fact that $h_{\mut}$ is constant on $ \partial \B (\rbulk)$, we have 
$$\int_{\B(\rbulk)} \diff \rv \:  \nablap h_{\mut} \cdot \nabla f = 0$$
by Stokes' formula, and thus
\[
\int_{\B(\rbulk)}  \diff \rv \: g ^2 |\jt |^2 = |\log \ep| ^2 \int_{\B(\rbulk)}  \diff \rv \: \frac{1}{g ^2 } |\nabla h_{\mut}| ^2 + \int_{\B(\rbulk)}  \diff \rv \: g^2 |\nabla f| ^2
\] 
and there only remains to drop the last positive term to conclude the proof.
\end{proof}

\subsection{Completion of the Lower Bound Proof}

Injecting \eqref{eq:low int energy} in \eqref{eq:low use balls}, we have 
\begin{multline*}
\int_{\B(R_+)} \diff \rv  \left\{\chiin \left( \half g ^2 |\nabla u| ^2 + |\log \ep| F \curl(iu,\nabla u) \right) \right\}\geq \left(1-C|\log \ep|^{-5}\right) |\log \ep| ^2 \int_{\B(\rbulk)}  \diff \rv \: \frac{1}{2g ^2} |\nabla h_{\mut}| ^2
\\ + \sum_{j\in J} 2\pi |\log \ep|\chiin(a_j)\left( \half  |d_i|g^2 (a_i) + d_i F(a_i)\right) \left(1 -C \frac{\log |\log \ep|}{|\log \ep|} \right) - \OO(|\log\eps| \log|\log\eps|)
\end{multline*}
and because of Proposition \ref{pro:g point est} we can replace $g^2$ and $F$ by $\tfm$ and $\tfpot$ respectively and absorb the new remainder in the old ones. Using in addition Proposition \ref{pro:jacest} we can also replace the sum over vortex balls and obtain  
\begin{multline}\int_{\B(R_+)} \diff \rv \left\{\chiin \left( \frac{1}{2}g^2 |\nabla u| ^2 + |\log \ep| F \curl(iu,\nabla u) \right) \right\}\geq \left(1-C|\log \ep|^{-5}\right) |\log \ep| ^2 \int_{\B(\rbulk)}  \diff \rv \: \frac{1}{2\tfm} |\nabla h_{\mut}| ^2
\\ + |\log \ep| ^2 \left(1 -C \frac{\log |\log \ep|}{|\log \ep|} \right) \int_{\B(\Rc)}  \chiin \left( \half \tfm  |\mu| +  \tfpot \mu \right) - \OO(|\log\eps| \log|\log\eps|).
\end{multline}
We can apply Lemma \ref{lem:regul vortic} with either $\phi = \chiin \tfm$ or $\phi = \chiin \tfpot$, to deduce
\begin{multline}
\int_{\B(R_+)} \diff \rv \: \chiin \left( \frac{1}{2}g ^2 |\nabla u| ^2 + |\log \ep| F \curl(iu,\nabla u) \right) \geq  
\left(1-C|\log \ep|^{-5}\right)  |\log \ep| ^2 \int_{\B(\rbulk)} \diff \rv \: \frac{1}{2\tfm} |\nabla h_{\mut}| ^2 \\
+  |\log \ep| ^2  \left(1 -C \frac{\log |\log \ep|}{|\log \ep|} \right) \int_{\B(\Rc)} \chiin  \left( \half \tfm  |\mut| +  \tfpot \mut \right) - \OO(|\log\eps| \log|\log\eps|).
\end{multline}
Now the definitions \eqref{eq:low pot} and \eqref{eq:low radii} guarantees that inside $\B(\Rc)\setminus \B(\rbulk)$
\[
|F| \leq \frac{1}{2} g ^2 
\]
by the same computation as in \eqref{est F outside}, and we can thus drop one more positive term  to arrive at (recall that $\chiin = 1$ in $\B(\rbulk)$)
\begin{multline}\label{eq:lowbound final-}
\int_{\B(R_+)}  \diff \rv \: \chiin \left( \frac{1}{2}g ^2 |\nabla u| ^2 + |\log \ep| F \curl(iu,\nabla u) \right) \geq  \\
|\log \ep| ^2  \left(1 -C \frac{\log |\log \ep|}{|\log \ep|} \right) \int_{\B(\rbulk)}  \lf( \frac{1}{2\tfm} |\nabla h_{\mut}| ^2 
+   \half \tfm  |\mut| +  \tfpot \mut   \ri) - \OO(|\log\eps| \log|\log\eps|).
\end{multline}
Gathering \eqref{eq:low decouple}, \eqref{eq:low cut domain 1}, \eqref{eq:low cut domain 2} and the above inequality, we thus conclude
\begin{multline}\label{eq:lowbound final}
\gpe \geq \hgpe +  |\log \ep| ^2 \left(1 -C \frac{\log |\log \ep|}{|\log \ep|} \right)  \int_{\B(\rbulk)} \left(\frac{1}{2\tfm} |\nabla h_{\mut}| ^2 
+  \frac{1}{2} \tfm  |\mut| +  \tfpot \mut   \right) \\	- \OO(|\log\eps| \log|\log\eps|)
\end{multline}
and one can recognize the renormalized energy $\tfIc$ of the regularized vorticity measure $\mut$ on the right-hand side, except that the domain is $\B(\rbulk)$ instead of $\B(\tfr)$. However it is clear from Theorem \ref{theo:renorm energy} that  $\B(\rrs) \subset \B(\rbulk)$ for any fixed $\Omega_0$ and $\ep$ small enough since $\rrs$ does not depend on $\ep$ and $\rbulk \to \tfr$ when $\ep \to 0$. We can thus invoke Remark \ref{sec:limit problem}.\ref{rem:domain} to deduce that the infimum of the renormalized energy appearing on the right-hand side of \eqref{eq:lowbound final} (taken with respect to $\mut$) is precisely $|\log \ep|^2\tfI$, so that one has 
\[
\gpe \geq \hgpe +  \tfI |\log \ep| ^2 - \OO\lf( |\log \ep| \log |\log \ep|\ri),
\]
which concludes the proof of Proposition \ref{pro:lower bound} and thus of Theorem \ref{teo:gse asympt}.


\subsection{Convergence of the Vorticity Measure}

It is now a short way to the proof of Theorem~\ref{teo:vorticity}. We apply the stability estimate \eqref{eq:stability I} to the renormalized energy appearing in \eqref{eq:lowbound final}, with integration domain restricted to $\B(\rbulk)$. Since such an energy has the same infimum and minimizer as $\tfIc$ (see again Remark \ref{sec:limit problem}.\ref{rem:domain}),
\[
\gpe \geq \hgpe +   |\log \ep| ^2\left( \tfI + \int_{\B(\Rc)} \frac{1}{2\tfm} |\nabla h_{\mut-\musta}| ^2 \right) \left(1 -C \frac{\log |\log \ep|}{|\log \ep|} \right) -\OO(\ep ^{\infty})
\]
where for any measure $\nu$ we denote now by $h_{\nu}$ the solution of
\beq
	\label{hnu}
\begin{cases}
-\nabla \lf( \frac{1}{\tfm} \nabla h_{\nu} \ri) = \nu,	&	\mbox{in } \B (\rbulk), \\
h_{\nu} = 0,							& 	\mbox{on } \dd \B(\rbulk).
\end{cases}
\eeq
Combining this with the energy upper bound of Proposition \ref{upper bound: pro}, we obtain
\begin{equation}
	\label{h1 convergence}
C \frac{\log |\log \ep|}{|\log \ep|^{1/2}}  \geq \int_{\B(\Rc)} \frac{1}{\tfm} |\nabla h_{\mut-\musta}| ^2 
\end{equation}
which clearly implies that
\begin{equation}\label{eq:estim mut}
 \sup_{\phi \in C^1_c (\B(\rbulk))} \frac{\bigg| \displaystyle\int_{\B(\rbulk)}  \left(\mut - \musta \right)\phi\bigg|}{\bigg(\displaystyle\int_{\B(\rbulk)}  \frac{1}{\tfm} |\nabla \phi| ^2\bigg) ^{1/2}} \leq C \left(\frac{\log |\log \ep|}{|\log \ep| ^{1/2}}\right) ^{1/2}.
\end{equation}
Indeed the estimate \eqref{h1 convergence} together with \eqref{hnu} imply that, for any differentiable function $ \phi $ with compact support strictly contained in $ \B(\Rc) $,
\bml{
 	\bigg| \displaystyle\int_{\B(\rbulk)}  \left(\mut - \musta \right)\phi\bigg| = \bigg|\displaystyle\int_{\B(\rbulk)}  \frac{1}{\tfm} \nabla h_{\mut - \musta} \cdot \nabla \phi \bigg|	\\
	 \leq C  \left(\frac{\log |\log \ep|}{|\log \ep| ^{1/2}}\right) ^{1/2} \bigg(\displaystyle\int_{\B(\rbulk)}  \frac{1}{\tfm} |\nabla \phi| ^2\bigg) ^{1/2},
}
by the Cauchy-Schwarz inequality. Note that we have also vindicated the claim contained in Remark \ref{sec:results}.\ref{rem:norm} that the regularized vorticity measure can be estimated in a better norm that $\mu$. 

Finally, combining \eqref{eq:estim mut} with Lemma \ref{lem:regul vortic}, we obtain
\[
 \sup_{\phi \in C^1_c (\B(\rbulk))} \frac{\bigg| \displaystyle\int_{\B(\rbulk)}  \left(\mu - \musta\right) \phi\bigg|}{\bigg(\displaystyle\int_{\B (\rbulk)} \frac{1}{\tfm} |\nabla \phi| ^2\bigg)^{1/2} + \left\Vert \nabla \phi \right\Vert_{L^{\infty}(\B(\rbulk)}} \leq C \left(\frac{\log |\log \ep|}{|\log \ep| ^{1/2}}\right) ^{1/2}
\]
where $\mu$ is defined in \eqref{eq:intro mu}.

\section*{Appendix}
\addcontentsline{toc}{section}{Appendix}

\renewcommand{\theequation}{A.\arabic{equation}}
\setcounter{equation}{0}
\setcounter{subsection}{0}
\renewcommand{\thesection}{A}

In this Appendix we collect some useful but technical estimates. Most of them have already been proven in very similar contexts so, for the sake of brevity, we keep the proofs as short as possible, referring to other papers for full details.

We start by investigating the properties of $ g $, i.e., the minimizer of $ \hgpf $. By standard arguments one can show that $ g $ is unique, smooth, positive, radial and decreasing and it solves the variational equation
\beq
	\label{g var eq}
	- g^{\prime\prime} - r^{-1} g^{\prime} + r^s g + 2 \eps^{-2} g^3 = \eps^{-2} \hchem g,
\eeq
for $ r \in \R^+ $. 

The evaluation of $ \hgpf $ on a suitable regularization of $ \sqrt{\tfm} $ as well as the trivial lower bound $ \hgpe \geq \tfe $ yield the estimate
\beq
	\label{hgpe asympt}
	\hgpe = \tfe + \OO(|\log\eps|),
\eeq
which in turn implies the convergence of $ g^2 $ to $ \tfm $ and a similar estimate of the difference between the chemical potentials:
\beq
	\label{g l2 est}
	\lf\| g^2 - \tfm \ri\|_{L^2(\R^2)} \leq \OO(\eps\sqrt{|\log\eps|}),	\qquad	\lf| \hchem - \tfchem \ri| \leq \OO(\eps\sqrt{|\log\eps|}).
\eeq
A very simple investigation of the variational equation \eqref{g var eq} shows that the maximum of $ g $ is bounded by the chemical potential $ \hchem $ and thus by a constant, i.e., 
\beq
	\label{g sup}
	\lf\| g \ri\|_{L^{\infty}(\R^2)} \leq \OO(1).
\eeq

The next result is an improvement of the estimate of the difference $ g^2 - \tfm $:

\begin{pro}[\textbf{Pointwise estimate of $ g^2 $}]
	\label{pro:g point est}
	\mbox{}	\\
	Let $ \Omega = \Om_0 |\log \ep|$ with $ \Omega_0 >0 $ fixed, then
	\beq
		\label{g point est}
		\sup_{r \in [0, \rtf - \eps^{2/3} |\log\eps|^{2/3}]} \lf| g^2(r) - \tfm(r) \ri| \leq \OO(\eps^{2/3} |\log\eps|^{2/3}).
	\eeq
\end{pro}

\begin{proof}
	The result has already been proven in a slightly different setting in \cite[Proposition 2.6]{CRY} (see also \cite{AAB,IM1} for similar results) but we repeat the crucial steps for convenience of the reader.
	
	The result is obtained through a local analysis of the variational equation \eqref{g var eq}, which can be rewritten in the form $ - \Delta g = 2 \eps^{-2} (\tilde{\rho} - g^2) g $, with
	\beq
		\tilde\rho(r) : = \half \lf[ \hchem - r^s \ri].
	\eeq
	Note that in the ball $ \B(\rtf - \frac{1}{2} \eps^{2/3} |\log\eps|^{2/3}) $ 
	one has the lower bound $ \tilde\rho(r) \geq C \eps^{2/3} |\log\eps|^{2/3} $, 
	since
	\beq
		\label{density diff}
		\lf|\tilde\rho(r) - \tfm(r)\ri| \leq \OO(\eps\sqrt{|\log\eps|}) \ll \tfm(r) = \OO(\eps^{2/3} |\log\eps|^{2/3}),	\quad	\mbox{if } r \leq \rtf - \half \eps^{2/3} |\log\eps|^{2/3},
	\eeq
	by \eqref{g l2 est}. 
	
	Now we consider an interval $ [r_0 - \delta, r_0 + \delta] $ with
	\bdm
		\delta = \half \eps^{2/3} |\log\eps|^{2/3},
	\edm
	for some $ \half \eps^{2/3} |\log\eps|^{2/3} \leq r_0 \leq \rtf - \eps^{2/3} |\log\eps|^{2/3} $. In this region it is possible to construct explicit super- (see \cite[Eq. (2.36)]{CRY} or \cite[Proof of Proposition 2.1]{AS}) and sub-solutions (see \cite[Eq. (2.47)]{CRY} or \cite{Ser}) to \eqref{g var eq}, which provide the estimates
	\bdm
		\tilde\rho\lf(r_0 + \delta \ri) \lf( 1 + \OO(\eps^{\infty}) \ri) \leq g^2(r_0) \leq \tilde\rho\lf(r_0 -  \delta \ri) \lf( 1 + \OO(\eps^{\infty}) \ri),
	\edm
	where we have also exploited the fact that $ \tilde{\rho} $ is decreasing. It remains now to replace $ \tilde\rho(r_0 \pm\delta) $ with $ \tilde\rho(r_0) $, which yields an error of order $ \eps^{2/3} |\log\eps|^{2/3} $ and use \eqref{density diff}. 
\end{proof}

Another important property of $ g $ is its exponential decay outside the support of $ \tfm $, i.e., for $ r > \rtf $:

\begin{pro}[\textbf{Exponential decay of $ g $}]
	\label{g exp decay: pro}
	\mbox{}	\\
	Let $ \Omega $ be as in Proposition \ref{pro:g point est}, then there exists two constants $ c > 0 $ and $ C < \infty $ such that for any $ r \geq \rtf + \eps^{2/3} $
	\beq
		\label{g exp decay}
		g^2(r) \leq C \eps^{2/3} |\log\eps|^{2/3} \exp \lf\{ - \frac{c(r - \rtf)}{\eps^{2/3}} \ri\}.
	\eeq
\end{pro}

\begin{proof}
	It suffices to notice that $ W(r) : = c_1 \eps^{1/3} |\log\eps|^{1/3} \exp \{ - c_2(r - \rtf)/\eps^{2/3} \} $ is a supersolution to \eqref{g var eq} for $ r \geq \rtf + \eps^{2/3} $ and for some constants $ c_1, c_2 $ with $ c_2 > 0  $ and $ c_1 < \infty $: Indeed at the boundary by \eqref{g point est}
	\bdm
		W(\rtf +\eps^{2/3}) = c_1 \eps^{1/3} |\log\eps|^{1/3} e^{-c_2} \geq g(r),	\quad \mbox{for } r \geq \rtf,
	\edm
	 if we pick $ c_1 = \OO(1) $ large enough. Moreover $ W $ satisfies $ - \Delta W + c_2^2 \eps^{-4/3} W \geq 0 $. On the other hand \eqref{g var eq} yields $ - \Delta g + C \eps^{-4/3} g \leq 0 $ because $ \tilde\rho \leq - C \eps^{2/3} $ in the same region, so that $ W $ is a supersolution for $ c_2 $ equal to $ \sqrt{C} > 0 $ and the result is proven. 
\end{proof}

We now turn to properties of $\gpm$, that we summarize in 

\begin{pro}[\textbf{Pointwise estimates of $ \gpm $}]
	\label{pro:gpm pointwise}
	\mbox{}	\\
	Let $ \Omega $ be as in Proposition \ref{pro:g point est}, then there exists two constants $ c > 0 $ and $ C < \infty $ such that for any $ r \geq \rtf + \eps^{2/3} $
	\beq
		\label{gpm exp decay}
		|\gpm(\rv)| ^2 \leq C \eps^{2/3} |\log\eps|^{2/3} \exp \lf\{ - \frac{c(r - \rtf)}{\eps^{2/3}} \ri\}.
	\eeq
	Moreover, there is a constant $C$ such that $|\gpm(\rv)| \leq C$ uniformly in $\R ^2$.
\end{pro}

\begin{proof}
	The proof is essentially the same as that of Proposition \ref{g exp decay: pro} plus some standard trick to extract a differential inequality for $ |\gpm| $. The reader can consult \cite[Proposition 2.5]{AAB}, \cite[Proposition 3.2]{IM1}  or \cite[Section 2.2]{CRY} where similar results are proved.
\end{proof}

\end{document}